\documentclass[submission,copyright,creativecommons]{eptcs}
\providecommand{\event}{AiML 2026} 

\usepackage{aiml26}

\usepackage{iftex}

\ifpdf
  \usepackage{underscore}         % Only needed if you use pdflatex.
  \usepackage[T1]{fontenc}        % Recommended with pdflatex
\else
  \usepackage{breakurl}           % Not needed if you use pdflatex only.
\fi

\usepackage[cal=euler,scr=boondox]{mathalfa}
\usepackage{amssymb}
\usepackage{mathtools}
\usepackage{bm}
\usepackage{stmaryrd}
\usepackage{multicol}
\usepackage{tikz}
	\usetikzlibrary{arrows.meta, shapes.geometric, positioning, patterns, calc}
\newcommand{\topo}{\mathcal{T}}

\newcommand{\powset}{\mathcal{P}}
\newcommand{\val}{\mathcal{V}}

\newcommand{\Int}{\mathit{Int}}
\newcommand{\Cl}{\mathit{Cl}}

\newcommand{\Prop}{\mathit{Prop}}
\newcommand{\A}{\mathrm{A}}
\newcommand{\E}{\mathrm{E}}

\newcommand{\f}{\varphi}
\newcommand{\ff}{\psi}

\newcommand{\tot}{\leftrightarrow}
\newcommand{\sem}[1]{\llbracket #1 \rrbracket}
\newcommand{\semM}[1]{\llbracket #1 \rrbracket_{\bm{M}}}
\newcommand{\Boxc}{\Box^{c}}
\newcommand{\Diamondc}{\Diamond^{c}}
\newcommand{\Lan}{\mathfrak{L}}
\newcommand{\LanBA}{\Lan^{\Box}_{\A}}
\newcommand{\LanBAc}{\Lan^{\Box}_{\A, c}}

\newcommand{\upset}{\mathcal{U}}

\newcommand{\sfAc}[1]{\mathsf{#1}_{\A, c}}
\newcommand{\itAc}[1]{\mathit{#1}_{\A, c}}
\newcommand{\itA}[1]{\mathit{#1}_{\A}}

\let\phi\f

\title{Non-classical Topological Evidence Logic}
\author{Igor Sedl\'{a}r
\institute{Institute of Computer Science, Czech Academy of Sciences\\ Prague, The Czech Republic}
\email{sedlar@cs.cas.cz}
}

\newcommand{\titlerunning}{Non-classical TEL}
\newcommand{\authorrunning}{I.\ Sedl\'{a}r}

\hypersetup{
  bookmarksnumbered,
  pdftitle    = {\titlerunning},
  pdfauthor   = {\authorrunning},
  pdfsubject  = {EPTCS},               % Consider adding a more appropriate subject or description
  pdfkeywords = {keyword1, keyword2} % Uncomment and enter keywords specific to your paper
}

\begin{document}
\maketitle

\begin{abstract}
Topological Evidence Logic (TEL) is a recent approach to epistemic logic that uses topological tools to model coherent epistemic justification. Specifically, a hypothesis is coherently justified if and only if it is entailed by a dense open set. In its simplest form, TEL can be formulated as an extension of $\mathsf{S4}$ with a global modality. 
All currently studied forms of TEL are based on classical propositional logic, which has been heavily criticised for misrepresenting the way in which ordinary agents reason. 
In this article, we show that the TEL approach is robust under modifications to the propositional base. We show that an extension of the intuitionistic modal framework recently introduced by de Groot and Shillito, incorporating a global modality, enables coherent justification to be expressed in an intuitionistic setting. 
Furthermore, we adapt the recent work of Standefer et al., which extends relevant logic with a global modality, to show that coherent justification can be expressed in a relevant setting if an interior-of-complement operator is added to the language.
Our main technical result is a soundness and completeness theorem for relevant TEL based on the weak relevant modal logic $\mathsf{BS4}$.
\end{abstract}

\section{Introduction}

% Context
Topology has proven to be a useful tool in logic and computer science. McKinsey and Tarski's famous results \cite{McKinsey1941,McKinseyTarski1944,Tarski1938} establish a \emph{topological semantics} for modal logic. Smyth's topological perspective on the basic concepts of programming language semantics \cite{Smyth1983} led to the modelling of \emph{observable properties} of computational systems by open sets \cite{Abramsky1991,Escardo2004,Smyth1992,Vickers1989}. This angle was further elaborated by Kelly \cite{Kelly1996} who provided a topological foundation of formal learning theory. 

In topological semantics of the modal logic $\mathsf{S4}$, the set of possible worlds $X$ is seen as a topological space and the proposition $\Box P$, for arbitrary $P \subseteq X$, is the \emph{interior} of $P$. Hence, $x \in \Box P$ iff there is an open set $U$ such that $x \in U$ and $U \subseteq P$. Assuming that open sets are interpreted as observable properties or \emph{evidence}, then $\Box P$ means that there is truthful evidence that supports $P$. 

Topological Evidence Logic (TEL) is a recent extension of the topological approach to modal logic that explicitly uses topology to model epistemic notions; see \cite{
BaltagEtAl2025a,BaltagEtAl2022a,BaltagEtAl2022,BaltagEtAl2016,
BaltagEtAl2017,BaltagEtAl2025,BjorndahlOzgun2019,FernandezGonzalez2018,Ozgun2017,OzgunEtAl2025} for example. The key concept in TEL is the topological representation of \emph{coherent justification}. Given a space $X$, a hypothesis $P \subseteq X$ is coherently justified if $U \subseteq P$ for a dense open set $U$. Intuitively, this means that $P$ is supported by evidence that is \emph{consistent with all (consistent) evidence}. Accordingly, $P$ is \emph{known} in a state $x \in X$ if it is coherently justified by an open set $U$ such that $x \in U$.  
  
A simple form of TEL is an extension of $\mathsf{S4}$ with the global modality $\A$, representing truth in all states. Since $\Box P$ is the interior of $P \subseteq X$, the proposition $\Diamond P$, defined as $(\Box (P^{c}))^{c}$, is the \emph{closure} of $P$. It follows that $\A \Diamond P = X$ iff $P$ is dense. Consequently, $\A \Diamond \Box P = X$ iff $P \subseteq U$ for some dense open set $U$. In other words, $\A \Diamond \Box P$ says that $P$ is supported by some dense open set, i.e.\ that $P$ is coherently justified.

% Gap
All currently studied forms of TEL are based on classical propositional logic and the use of Boolean negation appears to be essential for TEL's ability to express coherent justification. However, classical propositional logic has been heavily criticised for misrepresenting how agents reason. This has been a major theme in 20th century logic, giving rise to a plethora of non-classical logics such as intuitionistic, many-valued and relevant logics.  
In recent years, many non-classical epistemic logics have emerged; see \cite{ArtemovProtopopescu2016,BilkovaEtAl2016,
Sedlar2015,SedlarVigiani2022,SedlarVigiani2024,
StandeferMares2025,StandeferEtAl2023}. By using a non-classical propositional base, these epistemic logics avoid certain problematic closure principles such as `if an agent believes the consequent of an implication, then they believe the implication' or `if an agent has inconsistent beliefs, then they believe everything'. 
However, these logics lack the attractive evidence-based approach of TEL and are not able to represent coherent justification in the way TEL does. %To the best of our knowledge, non-classical versions of TEL have not been studied yet. 

% Contributions
In this article, we bridge this gap by \emph{studying intuitionistic and relevant versions of TEL.} 
Recent work by de~Groot and Shillito \cite{deGrootShillito2025} marks the first step towards a non-classical version of TEL. They extend the Kripke-style relational semantics for propositional intuitionistic logic with the interior semantics for the modal operator $\Box$. 
However, their logic $\mathsf{iS4}$ lacks the resources to express the density-based notion of coherent justification that is central to TEL.  
Our first contribution is to show that coherent justification can be expressed in an extension of $\mathsf{iS4}$ with $\A$. %We also prove soundness and strong completeness of this extension, although this result is not entirely new; see \cite{Ono1977}.

In the context of relevant logic, systems with $\A$ have recently been considered by Standefer and co-authors \cite{Standefer2022, StandeferFrench2025, StandeferMares2025}. However, completeness results for these logics have remained elusive. 
Our second contribution is to show that density-based justification can be expressed within an extension of the a semantic framework for the weak relevant modal logic $\mathsf{BS4}$, in which the $\Box$ operator is interpreted as interior and in which an interior-of-complement operator $\Boxc$ and the global modality $\A$ are introduced. 
We also prove soundness and completeness of the resulting logic. Due to the modular nature of our axiomatisation, our result solves the problem of axiomatising relevant logic with $\A$, left open in \cite{StandeferFrench2025}. 

%% Related work
% ConiglioPrietoSanabria2017
% Baskent2016
% ArtemovNogina2008 ?

% Overview
The article is structured as follows. Section \ref{sec:InTEL} discusses intuitionistic TEL and shows that coherent justification can be expressed in an extension of de Groot and Shillito's $\mathsf{iS4}$ with a global modality. in Section \ref{sec:InTELax}, we provide a sound and complete axiomatisation for intuitionistic TEL. Section \ref{sec:TowardReTEL} formulates a first naive attempt at relevant TEL and shows that it fails. Relevant TEL is introduced in Section \ref{sec:ReTEL} as an extension of the relevant modal logic $\mathsf{BS4}$ with the global modality $\A$ and the interior-of-complement operator $\Boxc$. A soundness and completeness result is established in Section \ref{sec:ReTELax}. Section \ref{sec:Conclusion} offers a summary of the paper and an  outline of future work. 
To ensure the paper remains self-contained for the non-expert reader, full proofs of two standard results in relevant logic are provided in the appendix.

\section{Intuitionistic TEL}\label{sec:InTEL}

In this section, we present an intuitionistic version of TEL, extending the work of de Groot and Shillito \cite{deGrootShillito2025}. We show that incorporating the global modality $\A$ into their logic $\mathsf{iS4}$ enables coherent justification to be expressed in a manner analogous to that in classical TEL. We assume that the reader is familiar with the basics of modal logic \cite{BlackburnEtAl2001} and general topology \cite{Munkres2000}.

\begin{definition}
Let $\Prop$ be a countably infinite set of propositional variables. The \emph{language $\LanBA$} is defined using the following grammar:
\[ \f \coloneq p \mid \neg \f \mid \f \land \f \mid \f \lor \f \mid \f \to \f \mid \Box \f \mid \A \f\]
where $p \in \Prop$. 
We define $\Diamond \f \coloneq \neg \Box \neg \f$ and $\E\f \coloneq \neg \A \neg \f$.
\end{definition}

\begin{definition}
Let $\langle X, \leq\rangle$ be a pre-ordered set. An \emph{up-set topology} on $\langle X, \leq\rangle$ is a topology $\topo$ on $X$ such that all $U \in \topo$ are closed upwards under $\leq$ (up-sets). An \emph{up-space} is $\langle X, \leq, \topo\rangle$ such that $\topo$ is an up-set topology on $\langle X, \leq\rangle$. 
Given an up-space $\bm{S} = \langle X, \leq, \topo\rangle$, the \emph{$\bm{S}$-interior} of $P \subseteq X$ is $\Int_{\bm{S}} (P) \coloneq \bigcup \{ U \mid U \in \topo \And U \subseteq P \}$ and the \emph{$\bm{S}$-closure} of $P$ is $\Cl_{\bm{S}}(P) \coloneq \bigcap \{ K \mid K^{c} \in \topo \And P \subseteq K \}$, where $K^{c}$ is the complement of $K$ with respect to $X$. 
\end{definition}

We will write $\Int$ and $\Cl$ if $\bm{S}$ is clear from the context or immaterial. 

\begin{example}\label{exam:alex}
Let $\langle X, \leq\rangle$ be a pre-ordered set. The \emph{Alexandroff topology} on $\langle X, \leq\rangle$ is $\mathcal{T}_{\leq}$ comprising all up-sets on $\langle X, \leq\rangle$. Then $\langle X, \leq, \mathcal{T}_{\leq} \rangle$ is an up-space. In fact, $\mathcal{T}_{\leq}$ is the \emph{finest} up-set topology on $\langle X, \leq \rangle$.
\end{example}

\begin{example}\label{exam:intrays}
Let $\langle \mathbb{R}, \leq\rangle$ be the set of real numbers with their standard ordering and choose a set $O \subseteq \mathbb{R}$ that intuitively represents the `observable' real numbers. We require that for each $a \in \mathbb{R}$ there is $b \in O$ such that $b \leq a$. Take $\mathcal{T}_{O}$ to be the right-order topology determined by the basis $\mathcal{B}_{O} = \{ [x, \infty) \mid x \in O \}$. The structure $\langle \mathbb{R}, \leq, \mathcal{T}_{O} \rangle$ is an up-space. 
One example of this construction is $\mathcal{T}_{\mathbb{Z}}$, representing possible measurements of a real quantity using a device with integer outputs.
\end{example}

\begin{example}[\cite{Smyth1992,Vickers1989}]\label{exam:words}
Let $2^{\infty}$ be the set of all finite and countably infinite binary words with the prefix ordering $\sqsubseteq$, where $w \sqsubseteq u$ if $w$ is a prefix of $u$. 
Let $W$ be a set of \emph{finite} binary words containing the empty word, representing possible observations. This set may contain all finite binary words, but this is not a requirement. We assume that the observer can observe some finite words but not infinite ones. 
Let ${\uparrow}w = \{ u \in 2^{\infty} \mid w \sqsubseteq u \}$, comprising words that are `consistent' with prefix $w$. Take the topology $\mathcal{T}_{W}$ generated by the basis $\mathcal{B}_{W} = \{ {\uparrow}w \mid w \text{ finite} \And w \in W \}$. 
Then $\langle 2^{\infty}, \sqsubseteq, \mathcal{T}_{W} \rangle$ is an up-space. Open sets in this topology are the observable (or \emph{semidecidable}) properties of binary words, i.e.\ properties that can be confirmed by observing a finite prefix.
\end{example}

\begin{example}[\cite{deGrootShillito2025,Ono1977}]\label{exam:ono}
Let $\langle X, \leq, R\rangle$ be a bi-relational Kripke model for intuitionistic modal logic \cite{Ono1977}, where $R$ is a preorder on $X$ such that $x \leq y$ implies $Rxy$. Let $\topo_R$ be the topology comprising all subsets of $X$ that are upwards closed under the preorder $R$ (the Alexandroff topology determined by $R$). Then $\langle X, \leq, \topo_R\rangle$ is an up-space. In fact, all up-spaces $\langle X, \leq, \mathcal{T} \rangle$ give rise to a bi-relational Kripke model $\langle X, \leq, \leq_{\mathcal{T}}\rangle$ where $\leq_{\mathcal{T}}$ is the \emph{specialisation order} of $\mathcal{T}$, in which $x \leq_{\mathcal{T}} y$ iff $\forall U \in \mathcal{T} \colon$ if $x \in U$, then $y \in U$. As discussed below, however, not all $\mathcal{T}$ are identical to $\mathcal{T}_{\leq_{\mathcal{T}}}$.
\end{example}

Relational structures based on pre-ordered sets $\langle X, \leq\rangle$ are often used in the semantics of non-classical logics. The usual informal interpretation of these structures is in terms of information support by partial situations: $x \leq y$ iff all information supported by situation $x$ is supported by situation $y$ as well \cite{Dunn1993, Grzegorczyk1964, Restall2005}. The Alexandroff topology $\mathcal{T}_{\leq}$ makes this informal interpretation explicit: its open sets are all the pieces of information that $\leq$ cares about since $x \leq y$ iff $\forall U \in \mathcal{T}_{\leq} \colon$ if $x \in U$, then $y \in U$.  
Conversely, any topological space $\langle X, \mathcal{T}\rangle$ can be transformed into a pre-ordered set $\langle X, \leq_{\mathcal{T}}\rangle$ where $\leq_{\mathcal{T}}$ is the specialisation order of $\mathcal{T}$ (see Example \ref{exam:ono}). However, a perfect match between the pre-order (implicit) and topological (explicit) model of available information cannot be expected in general. 
The Alexandroff topology determined by the specialisation order $\leq_{\mathcal{T}}$ of an arbitrary topology $\mathcal{T}$ is the closure of $\mathcal{T}$ under arbitrary intersections, not necessarily $\mathcal{T}$ itself. However, if $\mathcal{T}$ is already closed under arbitrary intersections ($\langle X, \mathcal{T}\rangle$ is an Alexandroff space), then $\mathcal{T}$ is the Alexandroff topology determined by $\leq_{\mathcal{T}}$. 

In a general up-space $\langle X, \leq, \mathcal{T}\rangle$ where $\mathcal{T}$ is typically coarser than $\mathcal{T}_{\leq}$, the pre-order $\leq$ is relative to \emph{all information} while $\mathcal{T}$ comprises the \emph{observable information}. The latter is a subset of the former and so open sets of $\mathcal{T}$ are up-sets with respect to $\leq$. 
In general, not all information is observable. For example, while $\mathcal{T}_{\leq}$ is closed under arbitrary intersections, $\mathcal{T}$ does not have to be -- an observer cannot usually make an infinite number of observations and combine them. 
However, $\mathcal{T}$ may differ from $\mathcal{T}_{\leq}$ even if it is closed under arbitrary intersections. In Example \ref{exam:intrays}, $a \in \mathbb{R}$ represent possible values of some quantity (e.g.\ temperature) and intervals $[x, \infty)$ for $x \in \mathbb{Z}$ represent temperature thresholds verifiable by a sensor with integer outputs. Collections of such intervals are closed under arbitrary intersections. 

\begin{definition}
Let $\upset (X, \leq)$ be the collection of up-sets on $\langle X, \leq\rangle$. An \emph{up-model} is $\langle X, \leq, \topo, \val\rangle$ where $\val : \Prop \to \upset (X, \leq)$. 
A \emph{pointed up-model} is a pair $\langle \bm{M}, x \rangle$ where $\bm{M}$ is an up-model on $X$ and $x \in X$. The \emph{satisfaction relation} $\models$ between pointed up-models and formulas $\chi \in \LanBA$ is defined by structural induction on $\chi$ as follows (where $\semM{\f} \coloneq \{ x \in X \mid \bm{M}, x \models \f \}$):
\begin{itemize}
\item $\bm{M}, x \models p$ iff $x \in \val(p)$;
\item $\bm{M}, x \models \neg \f$ iff $y \notin \semM{\f}$ for all $y \geq x$;
\item $\bm{M}, x \models \f \land \ff$ iff $\bm{M}, x \models \f$ and $\bm{M}, x \models \ff$;
\item $\bm{M}, x \models \f \lor \ff$ iff $\bm{M}, x \models \f$ or $\bm{M}, x \models \ff$;
\item $\bm{M}, x \models \f \to \ff$ iff, for all $y \geq x$, if $\bm{M}, y \models \f$, then $\bm{M}, y \models \ff$;
\item $\bm{M}, x \models \Box \f$ iff $x \in \Int_{\bm{M}} (\semM{\f})$;
\item $\bm{M}, x \models \A \f$ iff $\semM{\f} = X$.
\end{itemize}
A formula $\phi$ is valid in an up-model $\bm{M} = \langle X, \leq, \mathcal{T}, \mathcal{V} \rangle$ iff $\semM{\phi} = X$. Validity in up-spaces and classes thereof is defined as usual.
\end{definition}

We will often write $x \models \f$ and $x \in \sem{\f}$ if $\bm{M}$ is clear from the context or immaterial. 
Recall that a set $P \subseteq X$ is \emph{dense} in $\langle X, \topo\rangle$ iff $\Cl(P) = X$. Equivalently, $P$ is dense iff $P \cap U \neq \emptyset$ for all $U \in \topo$ such that $U \neq \emptyset$. %More discussion in \cite{BaltagEtAl2022}. 
Note that if $P$ is dense and $P \subseteq Q$, then $Q$ is dense.

Thanks to the fact that the \emph{interior of complement} can be expressed using the interior operator and intuitionistic negation, support by a dense open set can be expressed in the same way as in classical TEL. 

\begin{lemma}\label{lem:IntcIntuit}
$\Int (\sem{\neg \f}) = \Int (\sem{\f}^{c})$.
\end{lemma}
\begin{proof}
We have $\Int (\sem{\neg \f}) \subseteq \Int (\sem{\f}^{c})$ since $\sem{\neg \f} \subseteq \sem{\f}^{c}$. Conversely, $\Int (\sem{\f}^{c}) \subseteq \sem{\neg \f}$ since all open sets are upwards closed under $\leq$. Hence, $ \Int (\sem{\f}^{c}) \subseteq \Int(\sem{\neg \f})$.
\end{proof}

\begin{proposition}\label{prop:DBKintuit}
The following hold for all pointed up-models $\bm{M}, x$ and all formulas $\f$:
\begin{enumerate}
\item $\bm{M}, x \models \A\Diamond \f$ iff $\semM{\f}$ is a dense set;
\item $\bm{M}, x \models \A\Diamond\Box \f$ iff $U \subseteq \semM{\f}$ for a dense open set $U$;
\item $\bm{M}, x \models \Box \f \land \A\Diamond\Box \f$ iff $U \subseteq \semM{\f}$ for a dense open set $U$ such that $x \in U$.
\end{enumerate}
\end{proposition}
\begin{proof}
1. $\sem{\Diamond\f} = X$ iff $\forall x \colon x \notin \Int (\sem{\neg \f})$ iff $\forall x \colon x \notin \Int (\sem{\f}^{c})$ (Lemma \ref{lem:IntcIntuit}) iff $\Cl(\sem{\f}) = X$ iff $\sem{\f}$ is dense.

2. By the first claim, $\sem{\Diamond\Box\f} = X$ iff $\Int (\sem{\f})$ is dense. But clearly $\Int (\sem{\f})$ is dense iff there is a dense $U \in \topo$ such that $U \subseteq \sem{\f}$ (since the collection of dense sets is closed under supersets).

3. By the second claim, we have $x \models \Box \f \land \A\Diamond\Box \f$ iff $x \in U$ for some open $U \subseteq \sem{\f}$ and there is a dense open $V$ such that $V \subseteq \sem{\f}$. It follows that $U \cup V$ is a dense open such that $x \in U \cup V \subseteq \sem{\f}$. Conversely, if $x \in U$ for some dense open $U \subseteq \sem{\f}$, then obviously $x \models \Box \f$ and $\Int(\sem{\f})$ is non-empty and dense, whence $x \models \A \Diamond \Box \f$.
\end{proof}

The upshot of this section is that TEL is easily extended to the intuitionistic setting. Take the usual Kripke-style semantics for intuitionistic propositional logic, add a topological interpretation of $\Box$ as interior and extend the framework with the global modality $\A$. Even though negation is non-classical, we can still express interior-of-complement and therefore density. 
In Section \ref{sec:TowardReTEL}, we consider TEL based on relevant logic. We will see that the situation is more interesting there. 
However, before proceeding to relevant TEL, we will provide an axiomatisation of intuitionistic TEL. 

\section{Axiomatisation of Intuitionistic TEL}\label{sec:InTELax}

To the best of our knowledge, there is no explicit axiomatisation of intuitionistic modal logic with $\A$ in the literature. As discussed in \cite[p.~1304]{StandeferFrench2025}, for example, intuitionistic propositional logic with $\A$ can be axiomatised using Ono's system $\mathit{L_4}$ from \cite{Ono1977}. Here, we extend this axiomatisation to fit the language with both $\A$ and $\Box$. 
While the completeness proof may be relatively unexciting for an expert reader, we include it for the sake of exposition.

\begin{definition}
Let $\itA{iS4}$ be the Hilbert-style system that extends a fixed Hilbert-style system for intuitionistic propositional logic with the axioms and inference rules shows in Figure \ref{fig:HilbertiS4A}. We write $\itA{iS4} \vdash \phi$ to express that $\phi$ is a theorem of $\itA{iS4}$.
\end{definition}

\begin{figure}
\begin{center}
$\heartsuit (\f \to \ff) \to (\heartsuit\f \to \heartsuit\ff)$ \qquad
$\heartsuit \f \to \f$ \qquad
$\heartsuit \f \to \heartsuit \heartsuit \f$ \qquad
$\dfrac{\f}{\heartsuit \f}$\\[3mm]
$\A \f \lor \A \neg \A \f$ \qquad
$\A \gamma \to \Box \A \gamma$ \qquad 
$\neg \A \gamma \to \Box \neg \A \gamma$
\end{center}\caption{Modal axioms and rules of the system $\itA{iS4}$, where $\heartsuit \in \{ \Box, \A \}$.}\label{fig:HilbertiS4A}
\hrulefill
\end{figure}

The proof of the following lemma is straightforward, so we will omit it.

\begin{lemma}\label{lem:SoundnessiS4A}
For all $\phi \in \LanBA$, if $\itA{iS4} \vdash \phi$, then $\phi$ is valid in all up-spaces.
\end{lemma}

Both modal operators clearly distribute over conjunctions, $\itA{iS4} \vdash \heartsuit (\phi \land \psi) \tot (\heartsuit \phi \land \heartsuit \psi)$, and so they are monotonic: if $\itA{iS4} \vdash \phi \to \psi$, then $\itA{iS4} \vdash \heartsuit \phi \to \heartsuit \psi$. We note that $\neg \A \phi \to \A \neg \A \phi$ is provable using the axiom $\A \phi \lor \A \neg \A \phi$, distributivity of $\land$ over $\lor$ and the fact that $\bot \lor \psi \to \psi$ is provable for each contradiction $\bot$.

 \begin{definition}
 An \emph{$\itA{iS4}$-theory} is $\Gamma \subseteq \LanBA$ such that (a) if $\itA{iS4} \vdash \f \to \ff$ and $\f \in \Gamma$, then $\ff \in \Gamma$ and (b) if $\itA{iS4} \vdash \f$, then $\f \in \Gamma$. An $\itA{iS4}$-theory $\Gamma$ is \emph{prime} if $\f \lor \ff \in \Gamma$ only if $\f \in \Gamma$ or $\ff \in \Gamma$. An $\itA{iS4}$-theory $\Gamma$ is \emph{non-trivial} iff $\Gamma \neq \LanBA$.
 \end{definition}
 
 We denote as $|\phi|$ the set of non-trivial prime $\itA{iS4}$-theories $\Gamma$ such that $\phi \in \Gamma$.
 
 \begin{lemma}\label{lem:CanSpaceInt}
 The collection of sets of the form $| \Box \phi |$ for $\phi \in \LanBA$ is a basis for an up-set topology on the set of all non-trivial prime $\itA{iS4}$-theories ordered by set inclusion. 
 \end{lemma}
 \begin{proof}
 A modification of the standard argument \cite{AielloEtAl2003} for classical $\mathit{S4}$. Let $\Gamma$ be any non-trivial prime $\itA{iS4}$-theory. Firstly, by Necessitation, $\Box (p \to p) \in \Gamma$. Hence, every $\Gamma$ belongs to a basic set. 
 Secondly, if $\Gamma$ belongs to an intersection of two basic sets, i.e.~$\Gamma \in | \Box \phi | \cap |\Box \psi|$ for some $\phi, \psi$, then $\Gamma \in |\Box (\phi \land \psi)|$ since $\itA{iS4} \vdash \Box \phi \land \Box \psi \to \Box (\phi \land \psi)$. Hence, it follows that $\Gamma$  belongs to a basic subset of the intersection the two basic sets. Each basic set $|\Box \phi|$ is clearly closed under $\subseteq$.
 \end{proof}
 
 \begin{definition}
 The \emph{$\itA{iS4}$-space} is $\langle X, \subseteq, \mathcal{T} \rangle$, where $X$ is the set of all non-empty non-trivial prime $\itA{iS4}$-theories and $\mathcal{T}$ is the topology generated by the basis comprising sets $| \Box \phi | = \{ x \in X \mid \Box \f \in \Gamma \}$.
 \end{definition}
 
It follows from Lemma \ref{lem:CanSpaceInt} that the $\itA{iS4}$-space is a well-defined up-space. It can be extended to an up-model $\bm{C}$ by defining $\mathcal{V}$ as usual, namely, $\mathcal{V}(p) \coloneq |p|$. We can prove the Truth Lemma, stating that $\bm{C}, x \models \phi$ iff $\phi \in x$, for all formulas without $\A$. However, as usual in modal logic, it is not the case that $\A \phi \in x$ for some $x \in X$ iff $\phi \in y$ for all $y \in X$. 
What we need to do is to \emph{anchor $\bm{C}$ in a fixed $w \in X$}, that is restrict $\bm{C}$ to prime theories that agree with $w$ on all formulas of the form $\A \phi$.

\begin{definition}
Let $w$ be any non-trivial prime $\itA{iS4}$-theory. The \emph{canonical up-model for $\itA{iS4}$ anchored in $w$} is $\bm{C}^{w} = \langle X^{w}, \subseteq^{w}, \mathcal{T}^{w}, \mathcal{V}^{w} \rangle$ where
\begin{itemize}
\item $\displaystyle X^{w} \coloneq \bigcap_{\A\f \in w} | \A\f| \, \cap \, \bigcap_{\A\ff \notin w} |\neg \A\ff|$ and $x \subseteq^{w} y$ iff $x \subseteq y$ and $x,y \in X^{w}$;
 \item $\topo^{w}$ is the subspace topology on $X^{w}$ derived from $\topo$ (i.e.\ $U \in \topo^{w}$ iff $\exists V \in \topo \colon U = X^{w} \cap V$);
 \item $\mathcal{V}^{w}(p) = |p| \cap X^{w}$.
 \end{itemize}
\end{definition}

\begin{lemma}\label{lem:SameAint}
$x \in X^{w}$ iff $\A \phi \in x \iff \A\phi \in w$ for all $\phi$.
\end{lemma}
\begin{proof}
If $x \in X^{w}$ and $\A \phi \in w$, then $\A \phi \in x$ by definition. If $\A\phi \notin w$, then $\A \neg \A \phi \in w$ using the axiom $\A \phi \lor \A \neg \A \phi$ of $\itA{iS4}$. Then $\A\neg \A \phi \in x$ by definition, whence $\neg \A \phi \in x$ using the axiom $\A \psi \to \psi$. Since $x$ is non-trivial and $\itA{iS4} \vdash \psi \land \neg \psi \to \chi$, we have $\A \phi \notin x$. 

Conversely, assume that $x \notin X^{w}$. Then either there is $\A \phi \in w$ such that $\A \phi \notin x$ and we are done, or there is $\A \phi \notin w$ such that $\neg \A \phi \notin x$. In the latter case, axiom $\A \phi \lor \A \neg \A \phi$ makes sure that $\A \neg \A \phi \in w$. However, $\A \neg \A \phi \notin x$ since that would entail $\neg \A \phi \in x$ using axiom $\A \psi \to \psi$.
\end{proof}

\begin{definition}
A pair $\langle \Delta, \nabla\rangle$ of subsets of $\LanBA$ is \emph{$\itA{iS4}$-independent} if there are no $\f_1, \ldots, \f_n \in \Delta$ and $\ff_1, \ldots, \ff_m \in \nabla$ such that $\itA{iS4} \vdash \bigwedge_{i = 1}^{n} \f_i \, \to\, \bigvee_{j = 1}^{m} \ff_j$.
\end{definition}

\begin{lemma}[Pair extension for $\itA{iS4}$]\label{lem:PEint}
If $\langle \Delta, \nabla\rangle$ is an $\itA{iS4}$-independent pair, then there is a prime $\itA{iS4}$-theory $\Delta'$ such that $\Delta \subseteq \Delta'$ and $\nabla \cap \Delta' = \emptyset$. If both $\Delta$ and $\nabla$ are non-empty, then $\Delta'$ is non-empty and non-trivial.
\end{lemma}
\begin{proof}
See Appendix \ref{app:PE}.
\end{proof}

\begin{lemma}\label{lem:CompactInt}
Let $\Gamma$ be a non-empty set of formulas and $\psi$ an arbitrary formula. If $\left( \bigcap_{\f \in \Gamma} |\f| \right) \subseteq |\psi|$, then there is a finite $\Delta \subseteq \Gamma$ such that $\vdash \left ( \bigwedge_{\delta \in \Delta} \delta \right ) \to \psi$.
\end{lemma}
\begin{proof}
Standard argument using Zorn's lemma, the fact that $\itA{iS4}$ is finitary and Lemma \ref{lem:PEint}.
\end{proof}

We define $|\chi|^{w} \coloneq |\chi |\cap X^{w}$.

\begin{lemma}[Truth Lemma]\label{lem:TruthInt}
For all $\chi \in \LanBA$, $|\chi|^{w} = \sem{\chi}_{\bm{C}^{w}}$.
\end{lemma}
\begin{proof}
Structural induction on $\chi$. The interesting cases are (i) $\chi = \phi \to \psi$, (ii) $\chi = \Box \phi$ and (iii) $\chi = \A \phi$.

(i) $\chi = \phi \to \psi$. Assume that $\phi \to \psi \in x$, $x \subseteq^{w} y$ and, using the induction hypothesis, $\phi \in y$. Then $\phi \to \psi \in y$ and $\psi \in y$ since $y$ is closed under modus ponens. 
Conversely, if $\phi \to \psi \notin x$, then there is $y$ such that $x \subseteq y$, $\phi \in y$ and $\psi \notin y$ by Lemma \ref{lem:PEint} since the pair $\langle x \cup \{ \phi \}, \{ \psi \} \rangle$ is independent. We prove that $y \in X^{w}$ as follows. If $\A \gamma \in w$, then $\A \gamma \in x$. Since $\A \gamma \to (\phi \to \A \gamma)$ is a theorem of $\itA{iS4}$, we obtain $\phi \to \A \gamma \in y$. Hence, $\A \gamma \in y$. 
If $\A \gamma \notin w$, then $\neg \A \gamma \in x$. Since $\neg\A \gamma \to (\phi \to \neg\A\gamma)$ is a theorem of $\itA{iS4}$, we obtain $\phi \to \neg\A \gamma \in y$. Hence, $\neg \A\gamma \in y$.
%% In the relevant case, this proof differs: the "existence claim" is more complicated and we cannot rely on the positive paradox of material implication -- we have to assume these particular instances as axioms.

(ii) $\chi = \Box \phi$. If $\Box \phi \in x$, then $x \in | \Box \phi |^{w} \in \mathcal{T}^{w}$. Using the axiom $\Box \phi \to \phi$, we obtain $|\Box \phi | \subseteq |\phi|$ and so $|\Box\phi|^{w} \subseteq |\phi|^{w}$. 
Conversely, assume that there is $U \in \mathcal{T}^{w}$ such that $x \in U \subseteq |\phi|^{w}$. We know that $U = V \cap X^{w}$ for some $V \in \mathcal{T}$. Without loss of generality, we may assume that $x \in |\Box\psi|\cap X^{w} \subseteq |\phi|$ for some specific $\psi$. 
By Lemma \ref{lem:CompactInt}, there are $\{ \A \chi_i \}_{i = 1}^{n} \subseteq w \cap x$ and $\{ \A \chi'_{j} \}_{j = 1}^{m} \subseteq w^{c}$ such that
\[\vdash \Box \psi \, \land \, \bigwedge_{i = 1}^{n} \A \chi_i \, \land \, \bigwedge_{j = 1}^{m} \neg \A \chi'_{j} \, \to \, \phi\]
and $\{ \neg \A \chi'_{j} \}_{j = 1}^{m} \subseteq x$. 
We can infer the following using axioms $\A \gamma \to \Box \A \gamma$, $\neg \A \gamma \to \Box \neg \A \gamma$ and the fact that $\Box$ is monotonic and distributes over conjunctions in $\itA{iS4}$:
\[\vdash \Box \psi \, \land \, \bigwedge_{i = 1}^{n} \A \chi_i \, \land \, \bigwedge_{j = 1}^{m} \neg \A \chi'_{j} \, \to \, \Box \phi \, .\] 
It follows that $\Box \phi \in x$.

(iii) $\chi =  \A \phi$. If $\A \phi \in x$, then $\A \phi \in w$ by Lemma \ref{lem:SameAint}. Hence, $\A\phi \in y$ for all $y \in X^{w}$. Using the axiom $\A \phi \to \phi$, we infer that $\phi \in y$ for all $y \in X^{w}$. 
Conversely, assume that $\A \phi \notin x$. We show that the pair $\langle \{ \A \psi \mid \A \psi \in w \} \cup \{ \neg \A \psi' \mid \A \psi' \notin w \}, \{ \phi \} \rangle$ is independent. If not, then we have
\[ \vdash \bigwedge_{i = 1}^{n} \A \psi_i \, \land\, \bigwedge_{j = 1}^{m} \neg \A \psi'_j \, \to \, \phi\] 
for some $\{ \A\psi_i \}_{i = 1}^{n} \subseteq w$ and $\{ \A \psi'_{j} \}_{j = 1}^{m} \subseteq w^{c}$. 
We infer the following using axiom $\A \psi \to \A \A \psi$, the theorem $\neg \A \psi \to \A \neg \A \psi$ and the fact that $\A$ is monotonic and distributes over conjunctions in $\itA{iS4}$:
\[ \vdash \bigwedge_{i = 1}^{n} \A \psi_i \, \land\, \bigwedge_{j = 1}^{m} \neg \A \psi'_j \, \to \, \A\phi \, .\] 
Since the antecedent of this implication is in $x$ (as $x \in X^{w}$), we conclude that $\A\f \in x$, which contradicts our assumption. Hence, the pair is independent and we have $y \in X^{w}$ such that $\phi \notin y$ by Lemma \ref{lem:PEint}.
\end{proof}
%% In the relevant setting, we don't have the K axiom for $\Box$ or $\A$, so we have to assume distribution over conjunctions as axiom -- as ususal un weak relevant modal logics.

\begin{theorem}\label{thm:iS4ACompleteness}
A formula $\phi$ is valid in all up-spaces iff it is provable in $\itA{iS4}$.
\end{theorem}
\begin{proof}
Soundness is established by Lemma \ref{lem:SoundnessiS4A}. Completeness: if $\not\vdash \phi$, then $\langle T, \{ \phi \} \rangle$ is independent, where $T$ is the set of all theorems of $\itA{iS4}$. Then there is a non-trivial prime $\itA{iS4}$-theory $w$ such that $\phi \notin w$. By Lemma \ref{lem:TruthInt}, $\phi$ is not satisfied in $w$ in $\bm{C}^{w}$. Hence, $\phi$ is not valid in all up-spaces.
\end{proof}

\section{Towards relevant TEL}\label{sec:TowardReTEL}

The results of Section \ref{sec:InTEL} suggest a strategy for obtaining relevant TEL: first, take ordered frames for relevant logics \cite{Restall2000,RoutleyMeyer1973,RoutleyEtAl1982} and endow them with an up-set topology; add a modal $\Box$ to the language and interpret it as interior; add the global modality $\A$. In effect, this is the framework of \cite{StandeferFrench2025}, extended with a topologically interpreted $\Box$.

Firstly, we will set up this framework in more detail. Secondly, we will show that it cannot express interior-of-complement and that it cannot express the fact that a given formula is supported by a dense open set. In the next section, we extend the framework so that these facts become expressible.  

\begin{definition}
A \emph{Routley--Meyer frame} \cite{RoutleyEtAl1982,Standefer2026} is $\bm{F} = \langle X, \leq, N, R, \,^{*}\rangle$ where $\langle X, \leq\rangle$ is a pre-ordered set, $N \subseteq X$ is closed upwards under $\leq$, $R$ is a ternary relation on $X$ and $\,^{*}$ is a unary operation on $X$ such that
\begin{itemize}\setlength{\itemindent}{5mm}
\item[(RM1)] if $x \in N$ and $Rxyz$, then $y \leq z$;
\item[(RM2)] for all $y$, there is $x \in N$ such that $Rxyy$;
\item[(RM3)] if $Rxyz$ and $x' \leq x$, $y' \leq y$, $z \leq z'$, then $Rx'y'z'$;
\item[(RM4)] if $x \leq y$, then $y^{*} \leq x^{*}$;
\item[(RM5)] $x = x^{**}$.
\end{itemize}
A \emph{Routley--Meyer up-space} is a pair $\langle \bm{F}, \topo\rangle$ where $\topo$ is an up-set topology on $\bm{F}$.
\end{definition}

The motivation for adopting the frame conditions (RM1--5) is explained after Definition \ref{def:RMupmodel}. Before we get to the definition, we give two examples of RM up-models.

\begin{example}\label{exam:Z}
Expanding Example \ref{exam:intrays}, take $\bm{F}_{\mathbb{R}} = \langle \mathbb{R}, \leq, N, R, - \rangle$ where $\leq$ is the standard ordering of real numbers, $N = [0, \infty )$, and $Rxyz$ iff $x + y \leq z$. 
$\bm{F}_{\mathbb{R}}$ can be endowed by the right order topology $\topo_{\mathbb{Z}} = \{ \emptyset, \mathbb{R} \} \cup \{ [x, \infty) \mid x \in \mathbb{Z} \}$.
\end{example}

\begin{example}\label{exam:Sigma*}
Let $\Sigma^{*}$ be the set of all finite words over an alphabet $\Sigma$, including the empty word $\epsilon$. Take $\bm{F}_{\Sigma^{*}} = \langle \powset (\Sigma^{*}), \subseteq, N, R, \,^{c}\rangle$, where $A \in N$ iff $\epsilon \in A$ and $R(A,B,C)$ iff $AB= \{ wu \mid w \in A \And u \in B \} \subseteq C$. 
$\bm{F}_{\Sigma^{*}}$ can be endowed with the topology of `observable properties', generated by the sub-basis comprising sets $U_w = \{ A \subseteq \Sigma^{*} \mid w \in A \}$ for all $w \in \Sigma^{*}$. Intuitively, the open sets in this topology are properties of languages that can be verified by observing finite words. 
\end{example}

%Recall that $\upset (X, \leq)$ is the set of up-sets on $\langle X, \leq\rangle$; this notation extends to pre-ordered relational structures such as RM frames.

\begin{definition}\label{def:RMupmodel}
A \emph{RM up-model} is $\bm{M} = \langle \bm{F}, \topo, \val\rangle$ where $\val: \Prop \to \upset(X, \leq)$. Satisfaction is defined as follows (as before, $\semM{\f} = \{ x \mid \bm{M}, x \models \f \}$):
\begin{itemize}
\item $\bm{M}, x \models p$ iff $x \in \val(p)$;
%\item $\bm{M}, x \not\models \bot$ for all $x \in X$;
\item $\bm{M}, x \models \f \land \ff$ iff $\bm{M}, x \models \f$ and $\bm{M}, x \models \ff$;
\item $\bm{M}, x \models \f \lor \ff$ iff $\bm{M}, x \models \f$ or $\bm{M}, x \models \ff$;
\item $\bm{M}, x \models \f \to \ff$ iff, for all $y,z$, if $Rxyz$ and $\bm{M}, y \models \f$, then $\bm{M}, z \models \ff$;
\item $\bm{M}, x \models \neg \f$ iff $\bm{M}, x^{*} \not\models \f$;
\item $\bm{M}, x \models \Box \f$ iff $x \in \Int_{\bm{M}}(\semM{\f})$;
\item $\bm{M}, x \models \A \f$ iff $\bm{M}, y \models \f$ for all $y \in X$.
\end{itemize}
A formula $\f$ is \emph{valid} in an RM up-model $\bm{M}$ iff $N \subseteq \semM{\f}$. Let $\mathsf{BS4}_{\A}$ be the logic of all RM up-spaces, defined as the set of all $\LanBA$-formulas valid in all RM up-models. 
\end{definition}

The next lemma, useful in proving validity of implicational formulas, provides a motivation for the frame conditions (RM1--4). The frame condition (RM5) is needed for the double negation law to hold.

\begin{lemma}\label{lem:SemanticDeductionThm}
$\f \to \ff$ is valid in $\bm{M}$ iff $\semM{\f} \subseteq \semM{\ff}$.
\end{lemma}
\begin{proof}
This holds thanks to (RM1), (RM2) and the fact that $\semM{\f}$ us an up-set for all $\f$. The latter fact holds for $\f = \ff \to \chi$ by (RM3) and for $\f = \neg \ff$ by (RM4).
\end{proof}

We illustrate the satisfaction clauses of Definition \ref{def:RMupmodel} in the spaces discussed in previous examples.

\paragraph{Example \ref{exam:Z}, continued.}\hspace*{-1em} In any model $\bm{M}$ on $\bm{F}_{\mathbb{R}}$, we have $x \models \f \to \ff$ iff, for all $y,z \in \mathbb{R}$, if $y \models \f$, then $x + y \models \ff$ (`adding' $\phi$ gives $\psi$). On the other hand, $x \models \neg \f$ iff $-x \not\models \f$. Not all `properties' are observable directly if our measurement device has only integer outputs. For example, if $\sem{p} = [0.5, \infty)$, then $\sem{\Box p} = [1, \infty)$, meaning that measuring $1$ is the closest we can get to verifying $p$.

\vspace*{-2mm}

\paragraph{Example \ref{exam:Sigma*}, continued.}\hspace*{-1em} Formulas express properties of languages. In particular, $A \models \f \to \ff$ if  $AB \models \ff$ for all $B$ such that $B \models \f$ (`concatenating with a $\f$-language gives a $\psi$-language'). Not every property of languages can be verified by observing finite words. E.g.\ since every non-empty open set contains both finite and infinite languages, $\sem{\Box p} = \emptyset$ if $\sem{p} = \{ A \subseteq \Sigma^{*} \mid A \text{ is infinite} \}$.

\smallskip

Now we show that, over RM up-spaces, $\LanBA$ lacks the expressivity for the salient TEL concepts.

\begin{proposition}\label{prop:inexpressible}
For any given $p \in \Prop$: 
\begin{enumerate}
\item there is no $\f(p) \in \LanBA$ such that, for all $\bm{M}$ we have $\bm{M}, x \models \f(p)$ iff $x \in \Int_{\bm{M}}(\semM{p}^{c})$;
\item there is no $\ff(p) \in \LanBA$ such that, for all $\bm{M}$ we have $\bm{M}, x \models \ff(p)$ iff there is a dense open $U$ such that $U \subseteq \semM{p}$.
\end{enumerate}
\end{proposition} 
\begin{proof}
1. We construct an RM up-model $\bm{M}$ such that $\semM{\f} \neq \Int(\semM{p}^{c})$ for all $\f$. See Figure \ref{fig:no-intc}, where $N = X = \{ 1, \ldots, 4 \}$, $\leq$ is the identity relation and $Rxyz$ iff $y = z$. The Routley star operation is indicated by the solid arrows. The topology is $\topo = \{ \emptyset, \{ 1, 2 \}, \{ 3, 4 \}, X \}$. Moreover, $\val(p) = \{ 1 \}$ and $\val(q) = \emptyset$ for all other $q \in \Prop$. 
Let $\mathcal{F} = \{ \emptyset, \{ 1 \}, \{ 1, 2 \}, \{ 1, 2, 4 \}, X \}$. We prove by structural induction that $\llbracket \f\rrbracket \in \mathcal{F}$ for all $\f \in \LanBA$. 
The claim is obvious for elements of $\Prop$. The cases of the induction step corresponding to $\land$ and $\lor$ are also clear. For $\neg$, note that $\neg \emptyset \coloneq \{ x \mid x^{*} \notin \emptyset \} = X$, $\neg \{ 1 \} = \{ 1, 2, 4 \}$, $\neg \{ 1, 2 \} = \{ 1, 2 \}$, $\neg \{ 1, 2, 4 \} = \{ 1 \}$ and $\neg X = \emptyset$. For $\to$, note that $\sem{\f \to \ff} = X$ if $\sem{\f} \subseteq \sem{\ff}$ and $= \emptyset$ otherwise. 
For $\Box$, note that $\mathcal{F}$ is closed under $\Int$. Specifically, $\Int(\{ 1 \}) = \emptyset$ and $\Int(\{ 1, 2, 4 \}) = \{ 1, 2 \}$. The final case of the induction step corresponding to $\A$ is trivial. 
Note, however, that $\Int (\sem{p}^{c}) = \{ 3, 4 \}$, but $\{ 3, 4 \} \notin \mathcal{F}$. 

2. We present a pair of RM up-models $\bm{M}_1$ and $\bm{M}_2$ such that: (i) $\sem{\f}_{1} = \sem{\f}_{2}$ for all $\f \in \LanBA$, and (ii) there is a dense open subset of $\llbracket p\rrbracket_1$ while there is no dense open subset of $\llbracket p\rrbracket_2$. The two models are shown in Figure \ref{fig:no-B}, with $\bm{M}_1$ on the left and $\bm{M}_2$ on the right. The models are defined in the same way as the model in Figure \ref{fig:no-intc}; they differ from it only in the valuation of $p$ ($\val_1(p) = \val_2(p) = \{ 1, 2 \}$) and their respective topologies. In $\bm{M}_1$, we have $\topo_1 = \{ \emptyset, \{ 1, 2 \}, X \}$ and in $\bm{M}_2$ we have $\topo_2 = \{ \emptyset, \{ 1, 2 \}, \{ 3 \}, \{ 1, 2, 3 \}, X \}$. 
Let $\mathcal{F} = \{ \emptyset, \{ 1, 2 \}, X \}$. Similarly as before, we can prove that $\sem{\f}_{i} \in \mathcal{F}$ for all $\f \in \LanBA$ and $i \in \{ 1, 2 \}$. (Note that $\mathcal{F}$ is closed under both $\Int_1$ and $\Int_2$.) Using this fact, we can easily show that $\sem{\f}_1 = \sem{\f}_2$ for all $\f \in \LanBA$. 
Note that $\{ 1, 2 \}$ is dense in $\topo_1$. Therefore, in $\langle X, \topo_1\rangle$, there is a dense open subset of $\sem{p}_1$. However, $\{ 1, 2 \}$ is not dense in $\topo_2$ and so there is no dense open subset of $\sem{p}_2$ in $\langle X, \topo_2\rangle$. 
Hence, there is no $\f \in \LanBA$ that can express the existence of a dense open subset $\sem{p}$.
\end{proof}

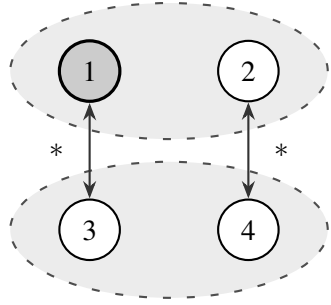
\begin{figure}[t]\centering
\begin{tikzpicture}[
	scale=0.7,
    %every node/.style={font=\sffamily},
    element/.style={circle, draw=black, thick, fill=white, minimum size=8mm},
    valuation/.style={circle, draw=black, very thick, fill=gray!40, minimum size=8mm},
    open/.style={ellipse, draw=black!70, loosely dashed, thick, fill=gray!15, minimum width=4.2cm, minimum height=1.8cm},
    involution/.style={<->, thick, draw=black!80, >=Stealth}
  ]
  
  % Background open sets (ellipses) drawn first so they sit behind the nodes
  \node[open] (U1) at (1.5, 1.5) {};
  %\node[anchor=north, text=black!80] at (4.5, 1.5) {\text{Open set }$\{1, 2\}$};
  
  \node[open] (U2) at (1.5, -1.5) {};
  %\node[anchor=south, text=black!80] at (4.5, -1.5) {\text{Open set }$\{3, 4\}$};

  % Domain elements
  % Node 1 is shaded darker with a thicker border to indicate the valuation
  \node[valuation] (n1) at (0, 1.5) {1}; %, label={[font=\bfseries, text=black!80]above right:$\sem{p}$}
  \node[element] (n2) at (3, 1.5) {2};
  \node[element] (n3) at (0, -1.5) {3};
  \node[element] (n4) at (3, -1.5) {4};
  
  % Routley star involution mappings (crossing the topological boundaries)
  \draw[involution] (n1) -- (n3) node[midway, left=2mm] {$*$};
  \draw[involution] (n2) -- (n4) node[midway, right=2mm] {$*$};

  % Overall space label
  %\node[font=\bfseries\large] at (1.5, -3.5) {Domain $X = \{1, 2, 3, 4\}$};

\end{tikzpicture}\caption{A model showing that interior of complement is not expressible in $\LanBA$.}\label{fig:no-intc}
\hrulefill
\end{figure}

\begin{figure}[t]\centering
\begin{tikzpicture}[
	scale=0.7,
    %every node/.style={font=\sffamily},
    element/.style={circle, draw=black, thick, fill=white, minimum size=8mm},
    valuation/.style={circle, draw=black, very thick, fill=gray!40, minimum size=8mm},
    open/.style={ellipse, draw=black!70, loosely dashed, thick, fill=gray!15, minimum width=3.5cm, minimum height=1.7cm},
    openSmall/.style={circle, draw=black!70, loosely dashed, thick, fill=gray!15, minimum size=1.4cm},
    involution/.style={<->, thick, draw=black!80, >=Stealth}
  ]

  % ================= MODEL 1 =================
  \begin{scope}[shift={(0,0)}]
    % Titles and Labels
    %\node[font=\bfseries\large] at (0, 3.5) {Model 1 ($\tau_1$)};
    %\node[font=\small, text=black!80] at (0, 2.9) {Target Property: \textbf{True} (Dense)};

    % Open sets (Background Layer)
    \node[open] (U112) at (0, 1) {};
    %\node[anchor=south, text=black!70, font=\small] at (0, 2.2) {Open set $\{1, 2\}$};

    % Domain Elements
    \node[valuation] (n11) at (-1.2, 1) {1};
    %, label={[font=\bfseries, text=black!80]left:$[p]$}
    \node[valuation] (n12) at (1.2, 1) {2};
    %, label={[font=\bfseries, text=black!80]left:$[p]$}
    \node[element] (n13) at (-1.2, -1.5) {3};
    \node[element] (n14) at (1.2, -1.5) {4};

    % Routley Star Involutions
    \draw[involution] (n11) -- (n13) node[midway, left=1mm] {$*$};
    \draw[involution] (n12) -- (n14) node[midway, right=1mm] {$*$};
  \end{scope}

  % ================= SEPARATOR =================
  \draw[thick, loosely dashed, draw=black!30] (3.5, -2.5) -- (3.5, 4);

  % ================= MODEL 2 =================
  \begin{scope}[shift={(7,0)}]
    % Titles and Labels
    %\node[font=\bfseries\large] at (0, 3.5) {Model 2 ($\tau_2$)};
    %\node[font=\small, text=black!80] at (0, 2.9) {Target Property: \textbf{False} (Not Dense)};

    % Open sets (Background Layer)
    \node[open] (U212) at (0, 1) {};
    %\node[anchor=south, text=black!70, font=\small] at (0, 2.2) {Open set $\{1, 2\}$};
    
    \node[openSmall] (U23) at (-1.2, -1.5) {};
    %\node[anchor=north, text=black!70, font=\small] at (-1.2, -2.4) {Open set $\{3\}$};

    % Domain Elements
    \node[valuation] (n21) at (-1.2, 1) {1};
    % , label={[font=\bfseries, text=black!80]left:$[p]$}
    \node[valuation] (n22) at (1.2, 1) {2};
    % , label={[font=\bfseries, text=black!80]right:$[p]$}
    \node[element] (n23) at (-1.2, -1.5) {3};
    \node[element] (n24) at (1.2, -1.5) {4};

    % Routley Star Involutions
    \draw[involution] (n21) -- (n23) node[midway, left=1mm] {$*$};
    \draw[involution] (n22) -- (n24) node[midway, right=1mm] {$*$};
  \end{scope}

\end{tikzpicture}\caption{A pair of models showing that support by a dense open set is not expressible in $\LanBA$.}\label{fig:no-B}
\hrulefill
\end{figure}
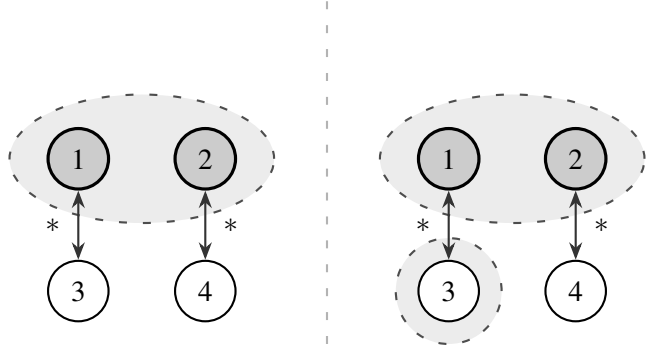

\section{Relevant TEL}\label{sec:ReTEL}

The previous section shows that $\LanBA$ needs to be extended in order to express the salient TEL concepts in the context of relevant logic. In this section, we introduce one such extension. It is perhaps the most direct one -- we add a modal operator $\Boxc$, interpreted as \emph{interior-of-complement}.

\begin{definition}
The \emph{language $\LanBAc$} is defined using the following grammar:
\[ \f \coloneq p \mid \neg \f \mid \f \land \f \mid \f \lor \f \mid \f \to \f \mid \Box \f \mid \A \f \mid \Boxc \f\]
where $p \in \Prop$. 
We define $\Diamondc \f \coloneq \neg \Boxc \f$ and $\E\f \coloneq \neg \A \neg \f$. We define $\bot \coloneq \A p \land \neg \A p$ and $\top \coloneq \A p \lor \neg \A p$ for some fixed $p \in \Prop$. 
\end{definition}

\begin{definition}
In RM up-models, the satisfaction relation $\models$ is defined for $\LanBAc$ as before, with the following clause added:
\begin{itemize}
\item $\bm{M}, x \models \Boxc \f$ iff $x \in \Int(\sem{\f}^{c})$.
\end{itemize}
Let $\mathsf{BS4}_{\A, c}$ be the logic of all RM up-models, defined as the set of all $\LanBAc$-formulas valid in all RM up-models. 
\end{definition}

Note that $\bm{M}, x \models \Diamond^{c} \f$ iff $x^{*} \in \Cl(\semM{\f})$. 
Moreover, $\bm{M}, x \models \E\f$ iff $\semM{\f} \neq\emptyset$ (for the right-to-left implication, we need to assume that every $y = z^{*}$ for some $z$), $\bm{M}, x \not\models \bot$ and $\bm{M}, x \models \top$ for all $x \in X$.

With $\Boxc$ at hand we can express density and coherent justification.

\begin{proposition}\label{prop:DBKrelevant}
The following hold for all pointed up-models $\bm{M}, x$:
\begin{enumerate}
\item $\bm{M}, x \models \A\Diamondc \f$ iff $\semM{\f}$ is a dense set;
\item $\bm{M}, x \models \A\Diamondc\Box \f$ iff $U \subseteq \semM{\f}$ for a dense open set $U$;
\item $\bm{M}, x \models \Box \f \land \A\Diamondc\Box \f$ iff $U \subseteq \semM{\f}$ for a dense open set $U$ such that $x \in U$.
\end{enumerate}
\end{proposition}
\begin{proof}
1. $\sem{\Diamondc \f} = X$ iff $\forall x \colon x^{*} \in \Cl(\sem{\f})$ iff $\forall y \colon y \in \Cl(\sem{\f})$ (since $\forall y \exists x \colon y = x^{*}$ by RM5) iff $\sem{\f}$ is dense. Claims 2.\ and 3.\ follow from Claim 1.\ exactly as in the proof of Proposition \ref{prop:DBKintuit}.
\end{proof}

This shows that $\mathsf{BS4}_{\A, c}$ is a version of relevant TEL that is able to express support by truthful evidence as $\Box \f$ and coherent justification by truthful evidence as $\Box \f \land \A \Diamondc \Box \f$.

We conclude this section with the following observation. Note that $\Boxc$ resembles intuitionistic negation under its topological interpretation, i.e.\ $\sem{\neg\f} = \Int (\sem{\f}^{c})$. One might wonder whether our language also allows us to express intuitionistic implication. The following proposition shows that it does not.

\begin{proposition}
For any given $p, q \in \Prop$, there is no $\f(p, q) \in \LanBAc$ such that, for all $\bm{M}$ we have $\semM{\f(p,q)} = \Int (\semM{p}^{c} \cup \semM{q})$. 
\end{proposition}  
\begin{proof}
We fix $p,q \in \Prop$ and we construct an RM up-model $\bm{M}$ such that $\semM{\f} \neq \Int (\semM{p}^{c} \cup \semM{q})$ for all $\f \in \LanBAc$. 
The countermodel is shown in Figure \ref{fig:no-intuitImp}. As before, $\leq$ is the identity relation, $N = X$, the $*$ operation is indicated by the arrows, and $Rxyz$ iff $y = z$. 
The topology $\topo$ is generated by the sub-basis $\{ B = \{ 1,2 \}, P = \{ 2, 3 \}, Q = \{ 2, 4 \} \}$. We set $\val (p) = P$ and $\val (q) = Q$. 
Observe that $\Int (P^{c} \cup Q) = \Int (\{ 1, 2, 4 \}) = \{ 1, 2, 4 \}$. However, we can show by structural induction that $\sem{\f} \in \{ \emptyset, X, \{ 2 \}, \{ 2, 3\}, \{ 2, 4 \}, \{ 2, 3, 4 \} \}$ for all $\f \in \LanBAc$.
\end{proof}

\begin{figure}\centering

\begin{tikzpicture}[
	scale=0.7,
    element/.style={circle, draw=black, thick, fill=white, minimum size=8mm},
    % Base style for open sets with transparency
    openSetBase/.style={draw=black!70, thick, loosely dashed, rounded corners=15pt, opacity=0.6},
    involution/.style={<->, thick, draw=black!80, >=Stealth}
  ]

  % --- COORDINATES (2x2 Grid) ---
  \coordinate (c1) at (0, 2); % Node 1
  \coordinate (c2) at (4, 2); % Node 2
  \coordinate (c3) at (0, 0); % Node 3
  \coordinate (c4) at (4, 0); % Node 4

  % --- OPEN SETS (Background Layer) ---
  
  % 1. The "Bridge" Set {1, 2} (Darkest Gray)
  % Rectangle around the top row
  \draw[openSetBase, fill=gray!40] 
    ($(c1)+(-1.2, 1.2)$) rectangle ($(c2)+(1.2, -1.2)$);

  % 2. Open Set p = {2, 3} (Medium Gray)
  % Diagonal shape connecting top-right and bottom-left
  \draw[openSetBase, fill=gray!25] 
    ($(c2)+(0.1, 1)$) -- ($(c3)+(-1.3, 0.2)$) -- ($(c3)+(-0.2, -1.4)$) -- 
    ($(c2)+(1.1, -0.3)$) -- cycle;
  % Label placed near node 3 for clarity
  \node[anchor=east] at (-1.6, 0) {$p$};

  % 3. Open Set q = {2, 4} (Lightest Gray)
  % Rectangle around the right column
  \draw[openSetBase, fill=gray!10] 
    ($(c2)+(-1, 1)$) rectangle ($(c4)+(1, -1)$);
  % Label placed near node 4 for clarity
  \node[anchor=west] at (5.1, 0) {$q$};

  % --- ELEMENTS (Foreground Layer) ---
  % Re-drawing nodes to sit on top of the gray fills
  \node[element] (n1) at (c1) {1};
  \node[element] (n2) at (c2) {2};
  \node[element] (n3) at (c3) {3};
  \node[element] (n4) at (c4) {4};

  % --- INVOLUTIONS ---
  \draw[involution] (n1) -- (n2) node[midway, above] {$*$};
  \draw[involution] (n3) -- (n4) node[midway, above] {$*$};

\end{tikzpicture}\caption{A model showing that intuitionistic implication is not expressible in $\LanBAc$.}\label{fig:no-intuitImp}
\hrulefill
\end{figure}
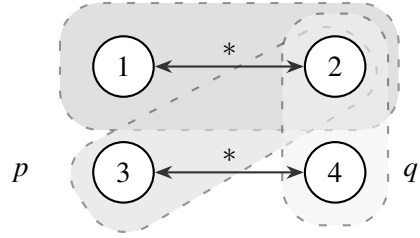

%In particular, the interpretation of $\Boxc\f$, the interior of the complement of $\sem{\f}$, resembles the interpretation of $\neg \f$ in the topological semantics of intuitionistic logic. This leads to the following observation.
%
%Let $\Lan$ be the propositional fragment of $\LanBAc$ and let $\LanBc$ be the fragment of $\LanBAc$ without $\A$. We define the following translation function $t : \Lan \to \LanBc$, reminiscent of the Gödel-McKinsey-Tarski translation:
%
%\begin{center}
%$t(p) = \Box p$ \qquad
%$t(\neg \f) = \Boxc t(\f)$ \qquad
%$t(\f \to \ff) = \Box (t(\neg \f) \lor t(\ff))$\\[2mm]
%$t (\f \land \ff) = t(\f) \land t(\ff)$\qquad
%$t (\f \lor \ff) = t(\f) \lor t(\ff)$
%\end{center} 
%
%We can show that, for all $\f \in \Lan \colon \f \in \mathsf{Int} \text{ iff } t(\f) \in \mathsf{BS4}_{c}$. By topological completeness of $\mathsf{Int}$, if $\f \notin \mathsf{Int}$, then there is an up-model $\bm{N} = \langle X, \leq, \topo, \val\rangle$ where $\leq$ can be assumed to be the identity relation and where $\val(p)$ is an open set  for each $p \in \Prop$. Let $[\f]_{\bm{N}}$ be the proposition expressed by $\f$ in this model, evaluated as usual in topological models for intuitionistic propositional logic. This model can be extended to a RM up-model $\bm{M}$ by adding $N, R$ and $\,^{*}$ in an arbitrary way. We can prove that $[\f]_{\bm{N}} = \semM{t(\f)}$ for all $\f \in \Lan$.  

\section{Axiomatisation of Relevant TEL}\label{sec:ReTELax}

In this section we provide a sound and complete axiomatisation of $\sfAc{BS4}$. This solves a problem left open in \cite{StandeferFrench2025}, namely, the problem of axiomatising $\mathsf{B}_{\A}$, the basic relevant logic $\mathsf{B}$ with $\A$. The axiomatisation of $\mathsf{B}_{\A}$ is obtained from our axiomatisation of $\sfAc{BS4}$ by omitting all axioms and rules that feature the modal operators $\Box$ and $\Boxc$.

Recall that $\mathsf{B}$ is the logic of all RM frames in the language $\Lan$, the purely propositional fragment of $\LanBAc$. A standard Hilbert-style system for $\mathsf{B}$, denoted as $\mathit{B}$, is shown in Figure \ref{fig:HilbertB}.

\begin{figure}[b]
\hrulefill\vspace*{-4mm}
\begin{center}
\begin{minipage}{0.48\linewidth}
%\vspace{0pt}
	\begin{gather}
	\f \to \f\label{B:Id}\\
	(\f \land \ff) \to \f, \quad (\f \land \ff) \to \ff\label{B:ConElim}\\
	\f \to (\f \lor \ff), \quad \ff \to (\f \lor \ff)\label{B:DisIntro}\\
	((\f \to \ff) \land (\f \to \chi)) \to (\f \to (\ff \land \chi))\label{B:ConIntro}\\
	((\f \to \chi) \land (\ff \to \chi)) \to ((\f \lor \ff) \to \chi)\label{B:DisElim}\\
	(\f \land (\ff \lor \chi)) \to ((\f \land \ff) \lor (\f \land \chi))\label{B:Distrib}
	\end{gather}
\end{minipage}	
\begin{minipage}{0.48\linewidth}
%\vspace{0pt}
	\begin{gather}
	\neg\neg \f \to \f\label{B:DNE}\\
	\f, \, \f \to \ff \implies \ff\label{B:MP}\\
	\f, \, \ff \implies \f \land \ff\label{B:Adj}\\
	\f \to \neg \ff \implies \ff \to \neg \f\label{B:Contra}\\
	\f \to \ff \implies (\chi \to \f) \to (\chi \to \ff)\label{B:Pref}\\
	\f \to \ff \implies (\ff \to \chi) \to (\f \to \chi)\label{B:Suff}
	\end{gather}
\end{minipage}		
\end{center}\caption{A Hilbert-style system $B$ for the basic relevant logic $\mathsf{B}$.}\label{fig:HilbertB}
\end{figure}

We note that the De Morgan laws $\neg (\f \land \ff) \tot (\neg \f \lor \neg \ff)$ and $\neg (\f \lor \ff) \tot (\neg \f \land \neg \ff)$ are provable in $B$, as is the double negation law $\neg\neg \f \tot \f$. The transitivity rule is derivable in $B$: if $B \vdash \f \to \ff$ and $B \vdash \ff \to \chi$, then $B \vdash \f \to \chi$. The contraposition rule is also derivable: if $B \vdash \f \to \ff$, then $B \vdash \neg \ff \to \neg \f$. 

\begin{definition}
Let $\itAc{BS4}$ be the Hilbert-style system that extends $B$ with the axioms and rules shown in Figure \ref{fig:HilbertBS4Ac}. Provability in $\itAc{BS4}$ is defined in the standard way. We write $\vdash \f$ if $\f$ is provable in $\itAc{BS4}$.
\end{definition}

\begin{figure}
\begin{center}
\begin{minipage}{0.48\linewidth}
%\vspace{0pt}
	\begin{gather}
(\heartsuit \f \land \heartsuit \ff) \to \heartsuit (\f \land \ff)\label{a:R}\\
\heartsuit \f \to \f\label{a:T}\\
\heartsuit \f \to \heartsuit \heartsuit \ff\label{a:4}\\
\f \to \ff \implies \heartsuit \f \to \heartsuit \ff\label{a:RuleM}\\
\neg\A \f \to \A \neg\A \f\label{a:A5}\\
\A (\f \to \ff) \to (\A \f \to \A \ff)\label{a:AK}\\
\Boxc\f \to \Box\Boxc\f\label{a:Boxc4}\\
\Boxc \f \land \f \to \bot\label{a:BoxcECQ}\\
\Box \ff \land \f \to \bot \implies \Box \ff \to \Boxc \f\label{a:BoxcRule}
	\end{gather}
\end{minipage}
\begin{minipage}{0.48\linewidth}
	\begin{gather}
(\A \f \land \neg \A\f) \to \ff \label{a:Aecq}\\
\f \to (\A \ff \lor \neg \A \ff)\label{a:Alem}\\
\A\f \to (\neg\A\f \to \ff)\label{a:AecqImp}\\
\neg\A\f \to (\A\f \to \ff)\label{a:AecqImpNeg}\\
%  \A\f \lor (\A\f \to \ff)\label{a:AlemOr}\\
\A\f \to (\ff \to \A\f)\label{a:Aweak}\\
\neg \A\f \to (\ff \to \neg\A\f)\label{a:AweakNeg}\\
\A\f \to \Box \A\f \label{a:ABoxA}\\
\neg\A\f \to \Box \neg\A\f \label{a:ABoxNegA}
	\end{gather}
\end{minipage}
\end{center}\caption{The extra axioms and rules of $\itAc{BS4}$, to be added to $\mathit{B}$, where $\heartsuit \in \{ \Box, \A \}$.}\label{fig:HilbertBS4Ac}
\hrulefill
\end{figure}

Axioms (\ref{a:R}--\ref{a:AK}) for $\A$ say that $\A$ is a relevant $\mathsf{S5}$-style modality; see \cite{Standefer2022}. 
(\ref{a:Boxc4}--\ref{a:BoxcRule}) say that $\Boxc \f$ expresses the largest open proposition disjoint from the proposition expressed by $\f$. 
Axioms (\ref{a:Aecq}--\ref{a:ABoxNegA}) impose \emph{uniformity} of formulas with the global modality under $*$ and $R$ and within open sets; note that, in the words of Standefer and French \cite{StandeferFrench2025}, a certain amount of \emph{classicality} is thereby imposed. These axioms are crucial for showing that the anchored canonical model is an RM up-model, as well as for proving the Truth Lemma for the model. 
We have not attempted to remove any redundancies from the axiom system.

As a first step towards our main result, we show that $\itAc{BS4}$ is sound with respect to the class of all RM up-spaces.

\begin{lemma}\label{lem:BS4AcSoundness}
If $\f$ is provable in $\itAc{BS4}$, then it is valid in all RM up-spaces.
\end{lemma}
\begin{proof}
Standard induction on the length of proofs. We show some cases.

\eqref{a:A5} Suppose $x \models \neg \A \f$; we have to prove that $x \models \A \neg \A \f$. That is, $y^{*} \not\models \A \f$ for an arbitrary $y$. But since $x^{*} \not\models \A \f$, there is $z \not\models \f$, which implies that $z' \not\models \A \f$ for all $z'$. 
\eqref{a:AK} Suppose $x \models \A (\f \to \ff)$. To prove that $x \models \A \f \to \A \ff$, assume that $Rxyz$ and $y \models \A \f$. Pick an arbitrary $u$. We have to prove $u \models \f$. By (RM2), there is $v \in N$ such that $Rvuu$. By the assumptions, $v \models \f \to \ff$ and $u \models \f$. Hence, $u \models \ff$, as desired. 
(\ref{a:Boxc4}--\ref{a:BoxcECQ}) hold since $\sem{\Boxc \f} = \Int (\sem{\f}^{c})$ is an open set $U$ such that $U \cap \sem{\f} = \emptyset$. \eqref{a:BoxcRule} holds since $\sem{\Boxc\f}$ is the largest such set.  
(\ref{a:Aecq}--\ref{a:Alem}) hold since $x \models \neg \A \f$ iff $x^{*} \not\models \A\f$ iff $\exists y \colon y \not\models \f$. Hence, $\sem{\neg\A\f} = \sem{\A\f}^{c}$. 
(\ref{a:AecqImp}--\ref{a:AweakNeg}) all hold because $\sem{\A\f}, \sem{\neg\A\f} \in \{ \emptyset, X \}$ and $(\emptyset \to P) = (Q \to X)$ for all $P,Q \subseteq X$ (where $Q \to P = \{ x \mid \forall y,z \colon \text{if } y \in Q \text{, then } z \in P \}$).  
(\ref{a:ABoxA}--\ref{a:ABoxNegA}) hold for a similar reason: since $\sem{\A\f}, \sem{\neg\A\f} \in \{ \emptyset, X \}$, they are open sets. 
\end{proof}

Comparing $\itAc{BS4}$ and $\itA{iS4}$, we note that the K-axiom $\Box (\phi \to \psi) \to (\Box \phi \to \Box\psi)$ is not valid and is thus replaced by the conjunctive regularity axiom $(\Box \phi \land \Box\psi) \to \Box (\phi \land \psi)$ and the monotonicity rule $\dfrac{\phi \to \psi}{\Box \phi \to \Box \psi}$. Because of a different definition of validity, we do not have the necessitation rule for $\Box$ or $\A$. Ono's axiom $\A \phi \lor \A \neg \A \phi$ is not sufficient for deriving $\neg \A \phi \to \A \neg \A \phi$, which is used in our completeness proof. Therefore, we replace it with $\neg \A \phi \to \A \neg \A \phi$ directly.

\begin{definition}
A \emph{$\itAc{BS4}$-theory} is $\Gamma \subseteq \LanBAc$ such that (a) if $\vdash \f \to \ff$ and $\f \in \Gamma$, then $\ff \in \Gamma$ and (b) if $\f \in \Gamma$ and $\ff \in \Gamma$, then $\f \land \ff \in \Gamma$. A $\itAc{BS4}$-theory $\Gamma$ is 
\emph{prime} iff $\f \lor \ff \in \Gamma$ only if $\f \in \Gamma$ or $\ff \in \Gamma$; 
\emph{regular} iff $\vdash \f$ only if $\f \in \Gamma$; and 
\emph{non-trivial} if $\Gamma \neq \LanBAc$.
\end{definition}

We denote as $|\phi|$ the set of all non-empty non-trivial prime $\itAc{BS4}$-theories that contain $\phi$.

 \begin{lemma}\label{lem:CanSpaceRM}
 The collection of sets of the form $| \Box \phi |$ for $\phi \in \LanBAc$ is a basis for an up-set topology on the set of all non-empty and non-trivial prime $\itAc{BS4}$-theories, ordered by set inclusion. 
 \end{lemma}
 \begin{proof}
 Fix a theory $x$. We know that $\phi \in x$ for some $\phi$. It follows that we have $\A\ff \in x$ or $\neg \A \ff \in x$ by \eqref{a:Alem}. In either case, we have $x \in |\Box \chi |$ for some $\chi$ by \eqref{a:ABoxA} or \eqref{a:ABoxNegA}. Hence, every $x$ is an element of some basic set. The rest of the proof is similar to the proof of Lemma \ref{lem:CanSpaceInt}. 
 \end{proof}

\begin{definition}
The \emph{$\itAc{BS4}$-space} is $\langle X, \subseteq, \topo\rangle$ where $X$ is the set of all non-empty non-trivial prime $\itAc{BS4}$-theories and $\topo$ is the topology generated by the basis comprising sets $| \Box \f |$ for all $\f \in \LanBAc$.
\end{definition}

The $\itAc{BS4}$-space can be extended to an RM up-model $\bm{C}$ where $N$, $R$, $\,^{*}$ and $\val$ are defined as usual in canonical models for relevant logics, namely by taking $N$ to be the set of regular prime theories, defining $Rxyz$ iff $\forall \f, \ff \colon$ if $\f \to \ff \in x$ and $\f \in y$, then $\ff \in z$; $x^{*} \coloneq \{ \f \mid \neg\f \notin x \}$; and $\val(p) \coloneq \{ x \in X \mid p \in x \}$.
However, the Truth Lemma does not hold in $\bm{C}$. Similarly as in the intuitionistic setting, it is not the case that $\A\f \in x$ iff $\f \in y$ for all $y \in X$. 
As in Section \ref{sec:InTELax}, we will solve this problem by anchoring $\bm{C}$ in a fixed $w \in X$ and restricting all the canonical relations to theories that agree with $w$ on all global modal formulas. We can assume that $w$ is a regular theory.

\begin{definition}\label{def:CanModAnchored}
Let $w$ be any regular non-trivial prime $\itAc{BS4}$-theory. The \emph{canonical up-model for $\itAc{BS4}$ anchored in $w$} is $\bm{C}^{w}= \langle X^{w}, \subseteq^{w}, \topo^{w}, N^{w}, R^{w}, \,^{*^{w}}, \val^{w}\rangle$ where:
\begin{itemize}
\item $\displaystyle X^{w} \coloneq \bigcap_{\A\f \in w} | \A\f| \, \cap \, \bigcap_{\A\ff \notin w} |\neg \A\ff|$ and $x \subseteq^{w} y$ iff $x \subseteq y$ and $x,y \in X^{w}$;
 \item $\topo^{w}$ is the subspace topology on $X^{w}$ derived from $\topo$ (i.e.\ $U \in \topo^{w}$ iff $\exists V \in \topo \colon U = X^{w} \cap V$);
 \item $N^{w} \coloneq \{ x \in X^{w} \mid x \text{ is regular}\}$
 \item $R^{w}xyz$ iff $x,y,z \in X^{w}$ and $\forall \f, \ff$, if $\f \to \ff \in x$ and $\f \in y$, then $\ff \in z$;
 \item $x^{*^{w}} \coloneq \{ \f \mid \neg\f \notin x \}$;
 \item $\val^{w}(p) \coloneq |p| \cap X^{w}$.
\end{itemize}
%\end{multicols}
We define $|\f|^{w} \coloneq |\f| \cap X^{w}$.
\end{definition}

\begin{lemma}\label{lem:CanWell}
$\bm{C}^{w}$ is well-defined.
\end{lemma}
\begin{proof}
We have to show that $x \in X^{w}$ implies $x^{*^{w}} \in X^{w}$. We have $x^{*} \in X$ by De Morgan and double negation laws. 
Now suppose that  $x^{*^{w}} \notin X^{w}$. If $\A\f \in w$ and $\A\f \notin x^{*}$, then $\A\f \land \neg\A\f \in x$ by \eqref{a:Alem}. By \eqref{a:Aecq}, this contradicts the assumption that $x$ is non-trivial. 
If $\neg \A\ff \notin x^{*}$ for some $\A\ff \notin w$, then $\neg\neg \A\ff \in x$, which entails $\A\ff \in x$ by \eqref{B:DNE}. However, $\neg\A\ff \in x$ as well, contradicting non-triviality of $x$ by \eqref{a:Aecq}.
\end{proof}

The following lemma establishes that $x \in X^{w}$ iff $x$ agrees with $w$ on all $\A\f$. 

\begin{lemma}
$x \in X^{w}$ iff $ \A\f \in x \iff \A\f \in w$ for all $\f$.
\end{lemma}
\begin{proof}
Firstly, assume that $x \in X^{w}$. If $\A\f \in w$, then $\A\f \in x$ by definition. If $\A\f \notin w$, then $\neg \A\f \in x$ by  definition and so $\A\f \notin x$ by \eqref{a:Aecq}. 
Secondly, assume that $x \notin X^{w}$. Either there is $\A\f \in w$ such that $\A\f \notin x$ and we are done, or there is $\A\ff \notin w$ such that $\neg\A\ff \notin x$. In the latter case, $\A\ff \in x$ by \eqref{a:Alem}.
\end{proof}

%We will make frequent use of the following variant of Lindenbaum's Lemma. It is a standard tool in relevant logic.

\begin{definition}
A pair $\langle \Delta, \nabla\rangle$ of subsets of $\LanBAc$ is \emph{$\itAc{BS4}$-independent} if there are no $\f_1, \ldots, \f_n \in \Delta$ and $\ff_1, \ldots, \ff_m \in \nabla$ such that $\vdash \bigwedge_{i = 1}^{n} \f_i \, \to\, \bigvee_{j = 1}^{m} \ff_j$.
\end{definition}

\begin{lemma}[Pair extension for $\itAc{BS4}$]\label{lem:PE}
If $\langle \Delta, \nabla\rangle$ is an $\itAc{BS4}$-independent pair, then there is a prime $\itAc{BS4}$-theory $\Delta'$ such that $\Delta \subseteq \Delta'$ and $\nabla \cap \Delta' = \emptyset$. If both $\Delta$ and $\nabla$ are non-empty, then $\Delta'$ is non-empty and non-trivial.
\end{lemma}
\begin{proof}
See Appendix \ref{app:PE}.
\end{proof}

\begin{lemma}\label{lem:CanModel}
$\bm{C}^{w}$ is an RM up-model.
\end{lemma}
\begin{proof}
(RM1) If $x \in N^{w}$, $R^{w}xyz$ and $\f \in y$, then $\f \in z$ since $\vdash \f \to \f$. Hence, $y \subseteq z$.

(RM2) Consider any $y \in X^{w}$. We have to show that there is $x \in N^{w}$ such that $R^{w}xyy$. Firstly, consider the pair $\langle T, \{ \f \to \ff \mid \f \in y \And \ff \notin y \}\rangle$, where $T$ is the set of theorems of $\itAc{BS4}$. This pair is independent. 
If it were not, then by adjunction \eqref{B:Adj} there would be a theorem $\chi$ such that $\vdash \chi \to \bigvee_{i = 1}^{n} (\f_i \to \ff_i)$ for $\f_i \in y$ and $\ff_i \notin y$. This entails that $\vdash \chi \to \big (\bigwedge_{i = 1}^{n} \f_i \to \bigvee_{i = 1}^{n} \ff_i \big)$, and so we obtain $\vdash \bigwedge_{i = 1}^{n} \f_i \to \bigvee_{i = 1}^{n} \ff_i$. Since $\bigwedge_{i = 1}^{n} \f_i \in y$, we have $\ff_i \in y$ for some $i$, contradicting the assumption. Hence, the pair must be independent. 
Since $y$ is non-empty and non-trivial, it follows by Lemma \ref{lem:PE} that there is a regular $x \in X$ such that if $\f \to \ff \in x$ and $\f \in y$, then $\ff \in y$. 

Now we prove that $x \in N^{w}$. Firstly, suppose that $\A\f \in w$ and $\A\f \notin x$. We have $\neg\A\f \in x$ by \eqref{a:Alem}. Since $\ff \notin y$ for some $\ff$, we obtain $\A\f \to \ff \in x$ by \eqref{a:AecqImpNeg}. Since $y \in X^{w}$, we have $\A\f \in y$. It follows that $\ff \in y$, contradicting our assumption. 
Secondly, suppose that $\A\ff \in x$ and $\A\ff \notin w$. We have $\f \in y$ for some $\f$ and so $\f \to \A\ff \in x$ by \eqref{a:Aweak}. Consequently, $\A\ff \in y$, contradicting the assumption that $y \in X^{w}$. 

(RM3) is obviously satisfied and the rest is standard: (RM4) is established using contraposition and (RM5) holds thanks to double negation.
\end{proof}

\begin{lemma}\label{lem:Compact}
Let $\Gamma$ be a non-empty set of formulas and $\psi$ an arbitrary formula. If $\left( \bigcap_{\f \in \Gamma} |\f| \right) \subseteq |\psi|$, then there is a finite $\Delta \subseteq \Gamma$ such that $\vdash \left ( \bigwedge_{\delta \in \Delta} \delta \right ) \to \psi$.
\end{lemma}
\begin{proof}
Standard argument using Zorn's lemma, the fact that $\itAc{BS4}$ is finitary and Lemma \ref{lem:PE}.
\end{proof}

\begin{lemma}\label{lem:TruthLemma}
For all $\chi$, $|\chi|^{w} = \sem{\chi}_{\bm{C}^{w}}$.
\end{lemma}
\begin{proof}
Structural induction on $\chi$. The only interesting cases are (i) $\chi = \f \to \ff$, (ii) $\chi = \Box\f$, (iii) $\chi = \Boxc \f$ and (iv) $\chi = \A\f$.

(i) $\chi = \f \to \ff$. If $\f \to \ff \in x$ and $R^{w}xyz$ where $\f \in y$, then $\ff \in z$ by the definition of $R^{w}$. 
Conversely, if $\f \to \ff \notin x$, then we can show that there are $y,z \in X$ such that $Rxyz$, $\f \in y$ and $\ff \notin z$; see Appendix \ref{app:TruthImp}. We prove that $y, z \in X^{w}$ as follows. 
Assume that $y \notin X^{w}$. Firstly, if there is $\A\gamma \in w$ such that $\A\gamma \notin y$, then we have $\A\gamma \in x$ by definition and $\neg\A\gamma \in y$ by \eqref{a:Alem}. Since $\ff \notin z$ for some $\ff$, we obtain $\neg\A\gamma \to \ff \in x$ by \eqref{a:AecqImp}, and so $\ff \in z$. This is a contradiction. 
Secondly, assume that there is $\A\gamma \notin w$ such that $\A\gamma \in y$. Since $\ff \notin z$ for some $\ff$, we obtain $\A\gamma \to \ff \in x$ by \eqref{a:AecqImpNeg}, and so $\ff \in z$. Contradiction. 
The fact that $z \in X^{w}$ is established similarly, using \eqref{a:Aweak} and \eqref{a:AweakNeg} and the fact that $y$ is non-empty.

(ii) $\chi = \Box \f$. This case is established similarly as in the proof of Lemma \ref{lem:TruthInt}, using Lemma \ref{lem:Compact} and \eqref{a:R}, \eqref{a:T}, \eqref{a:4}, \eqref{a:RuleM}, \eqref{a:ABoxA}, \eqref{a:ABoxNegA}. 
% \mathit{If $\Box \f \in x$, then $x \in U = | \Box \f |^{w}$. Clearly $U \in \topo^{w}$. By \eqref{a:T}, $|\Box \f| \subseteq |\f|$ and so $|\Box \f|^{w} \subseteq |\f|^{w}$. 
%Conversely, assume that there is $U \in \topo^{w}$ such that $x \in U$ and $U \subseteq |\f|^{w}$. We know that $U = V \cap X^{w}$ for some $V \in \topo$. Without loss of generality, we may assume that $x \in |\Box \ff| \cap X^{w} \subseteq |\f|$. By Lemma \ref{lem:Compact}, there are $\{ \A \chi_i \}_{i = 1}^{n} \subseteq w \cap x$ and  $\{ \A \chi'_j \}_{j = 1}^{m} \subseteq w^{c}$ such that 
%\[ \vdash \Box \ff \, \land \, \bigwedge_{i = 1}^{n} \A\chi_i \,\land\, \bigwedge_{j = 1}^{m} \neg \A\chi'_j \,\to\, \f \]
%and $\{ \neg\A \chi'_j \}_{j = 1}^{m} \subseteq x$. We may infer the following using \eqref{a:R}, \eqref{a:4}, \eqref{a:RuleM}, \eqref{a:ABoxA} and \eqref{a:ABoxNegA}:
% \[ \vdash \Box \ff \, \land \, \bigwedge_{i = 1}^{n} \A\chi_i \,\land\, \bigwedge_{j = 1}^{m} \neg \A\chi'_j \,\to\, \Box\f \, .\] We know that the antecedent of this implication is in $x$, and so we conclude that $\Box \f \in x$, as desired.} 
 
 (iii) $\chi \in \Boxc \f$. If $\Boxc \f \in x$, then $x \in U = |\Box\Boxc\f|^{w}$ by \eqref{a:Boxc4}. Obviously $U \in \topo^{w}$. Moreover, $U \subseteq |\Boxc\f|^{w} \subseteq (|\f|^{w})^{c}$ by \eqref{a:T} and \eqref{a:BoxcECQ}.  Hence, $x \in \Int ((|\f|^{w})^{c})$, as desired. 
 Conversely, assume that $x \in U$ such that $U \subseteq (|\f|^{w})^{c}$. Without loss of generality, we may assume that $x \in |\Box \psi|^{w}$ for some $\psi$ such that $|\Box \psi|^{w} \subseteq (|\f|^{w})^{c}$. In other words, $|\Box \psi| \cap |\f| \cap X^{w} = \emptyset$. By Lemma \ref{lem:Compact},  there are $\{ \A \chi_i \}_{i = 1}^{n} \subseteq w \cap x$ and  $\{ \A \chi'_j \}_{j = 1}^{m} \subseteq w^{c}$ such that 
\[ \vdash \Box \ff \, \land \, \f \, \land \, \bigwedge_{i = 1}^{n} \A\chi_i \,\land\, \bigwedge_{j = 1}^{m} \neg \A\chi'_j \,\to\, \bot \, .\] By \eqref{a:T} and \eqref{a:RuleM}, we obtain
 \[ \vdash \Box \Big ( \ff \, \land \, \bigwedge_{i = 1}^{n} \A\chi_i \,\land\, \bigwedge_{j = 1}^{m} \neg \A\chi'_j \Big) \, \land \, \f \,\to\, \bot \, , \]
 which we denote by $\vdash \Box \alpha \land \f \to \bot$. Using \eqref{a:BoxcRule}, $\vdash \Box \alpha \to \Boxc \f$. We know that $\Box \alpha \in x$, and so we conclude that $\Boxc \f \in x$, as desired.
 
 (iv) $\chi = \A\f$. This case is established similarly as in the proof of Lemma \ref{lem:TruthInt}, using Lemma \ref{lem:PE} and \eqref{a:R}, \eqref{a:T}, \eqref{a:4}, \eqref{a:RuleM} and \eqref{a:A5}. 
% If $\A\f \in x$, then $\A\f \in w$, and so $\A\f \in y$ for all $y \in X^{w}$. By \eqref{a:T}, $\f \in y$ for all $y \in X^{w}$. 
% Conversely, assume that $A\f \notin x$. We show that the pair $\langle \{ \A\ff \mid \A\ff \in w \} \cup \{ \neg \A\ff' \mid \A\ff' \notin w \}, \{ \f \}\rangle$ is independent. If not, then we have
% \[ \vdash \bigwedge_{i = 1}^{n} \A\ff_i \, \land \, \bigwedge_{j = 1}^{m} \neg \A\ff'_j \, \to\, \f\] for some $\{ \A \psi_i \}_{i = 1}^{n} \subseteq w$ and  $\{ \A \psi'_j \}_{j = 1}^{m} \subseteq w^{c}$. 
% By \eqref{a:R}, \eqref{a:4}, \eqref{a:A5} and \eqref{a:RuleM}, we obtain 
%  \[ \vdash \bigwedge_{i = 1}^{n} \A\ff_i \, \land \, \bigwedge_{j = 1}^{m} \neg \A\ff'_j \, \to\, \A\f \, .\] Since the antecedent of this implication is in $x$ (as $x \in X^{w}$), we conclude that $\A\f \in x$, which contradicts our assumption. Hence, the pair is independent and we obtain a $y \in X^{w}$ such that $\phi \notin y$ by Lemma \ref{lem:PE}.
\end{proof}

\begin{theorem}\label{thm:BS4AcCompleteness}
A formula $\f$ is valid in all RM up-spaces iff it is provable in $\itAc{BS4}$.
\end{theorem}
\begin{proof}
Soundness is established in Lemma \ref{lem:BS4AcSoundness}. 
Completeness: if $\not\vdash \phi$, then $\langle T, \{ \phi \}\rangle$ is independent. Then there is a regular prime theory $w$ such that $\phi \notin w$. By Lemmas \ref{lem:CanWell}, \ref{lem:CanModel} and \ref{lem:TruthLemma}, $\phi$ is not valid in all RM up-models.
\end{proof}

\section{Conclusion}\label{sec:Conclusion}

This paper made the first steps in studying versions of Topological Evidence Logic (TEL) based on non-classical propositional logics. Specifically, we studied intuitionistic and relevant versions of TEL. 
We have shown that an extension of de Groot and Shillito's \cite{deGrootShillito2025} modal logic $\mathsf{iS4}$ with a global modality $\A$ can express TEL-style coherent justification (existence of a dense open subset). We have argued that a language with $\Box$ and $\A$ is not sufficient for a relevant version of TEL since coherent justification is not expressible. 
Accordingly, for relevant TEL we added an interior-of-complement operator $\Boxc$ and we showed that coherent justification is expressible in the language even in a weak relevant background. 
As our main technical result, we provided a sound and complete axiomatisation of a relevant TEL based on the weak modal relevant logic $\mathsf{BS4}$. We have also provided a sound and complete axiomatisation of intuitionistic TEL based on $\mathsf{iS4}$ with $\A$.

The paper barely scratches the surface, leaving many interesting problems to be investigated in the future. Firstly, both $\mathsf{iS4}$ and $\mathsf{BS4}$ are decidable and we plan to extend the decidability result to $\mathsf{iS4}_{\A}$ and $\sfAc{BS4}$. It should not be too difficult to combine proofs of the finite model property with the anchoring technique. 
Secondly, it will be interesting to look at relevant TEL based on stronger propositional logics such as $\mathsf{R}$ or $\mathsf{E}$. It is not clear how much the extra frame conditions needed for these logics interfere with the anchoring technique. 
Thirdly, it is intriguing in the setting of RM frames to replace topological spaces seen as locales of up-sets, where evidence combination is represented by set intersection, by \emph{quantales} of up-sets where evidence combination is represented by the ternary relation $R$, and investigate the corresponding modal logic.  

\paragraph{Acknowledgement.} This work was supported by the Czech Science Foundation grant no.\ 25-17958J. The author is grateful to the reviewers for valuable feedback and suggestions that helped to improve the paper. 

\appendix

\section{Technical appendix}

\subsection{Proof of the Pair Extension Lemma}\label{app:PE}

We prove the lemma in a slightly more general setting that entails both Lemma \ref{lem:PEint} and Lemma \ref{lem:PE}; see \cite[pp.~92--95]{Restall2000}.  
Let $L$ be any extension of $B$ with extra axiom schemata and rules, over any language $\Lan'$ extending $\Lan$ such that $Fm_{\Lan'}$ is denumerable. We express the fact that $\phi$ is provable in $L$ by writing $\vdash \phi$. 
As expected, a pair $\langle \Delta, \nabla\rangle$ of subsets of $\Lan'$ is $L$-independent iff there are no $\f_1, \ldots, \f_n \in \Delta$ and $\ff_1, \ldots, \ff_m \in \nabla$ such that $\vdash \bigwedge_{i = 1}^{n} \f_i \, \to\, \bigvee_{j = 1}^{m} \ff_j$. 
(Prime, regular) $L$-theories are defined as expected. 

Note that transitivity is a derived rule in $L$: if $\vdash_{L} \phi \to \psi$ and $\vdash_{L} \psi \to \chi$, then applying the suffixing rule \eqref{B:Suff} to $\vdash_L \phi \to \psi$ yields $\vdash (\psi \to \chi) \to (\phi \to \psi)$. Because we have $\vdash \psi \to \chi$, applying modus ponens \eqref{B:MP} yields $\vdash \phi \to \chi$.

\begin{lemma}[Pair Extension]\label{lem:PExGeneral}
If $\langle \Delta, \nabla\rangle$ is an $L$-independent pair of subsets of $\Lan'$, then there is a prime $L$-theory $\Delta_{\infty}$ such that $\Delta \subseteq \Delta_{\infty}$ and $\Delta_{\infty} \cap \nabla = \emptyset$. 
$\Delta_{\infty}$ is non-empty and non-trivial if $\Delta$ and $\nabla$ are both non-empty.
\end{lemma}

To prove Lemma \ref{lem:PExGeneral}, we first require a step lemma demonstrating that $L$-independence can be preserved when adding a new formula to either the left or right set.

\begin{lemma}\label{lem:StepGeneral}
For all $\chi \in Fm_{\mathfrak{L}'}$, if $\langle \Delta, \nabla\rangle$ is $L$-independent, then $\langle \Delta \cup \{ \chi \}, \nabla\rangle$ is $L$-independent or $\langle \Delta, \nabla \cup \{ \chi \}\rangle$ is $L$-independent.
\end{lemma}

\begin{proof}
Suppose that both $\langle \Delta \cup \{ \chi \}, \nabla\rangle$ and $\langle \Delta, \nabla \cup \{ \chi \}\rangle$ are not $L$-independent. It follows that there are conjunctions $\phi, \phi'$ of elements of $\Delta$ and disjunctions $\psi, \psi'$ of elements of $\nabla$ such that
\[ \vdash_{L} (\phi \land \chi) \to \psi \qquad \text{and} \qquad \vdash \phi' \to (\chi \lor \psi') \, .\] 
(We use $\land$-elimination \eqref{B:ConElim} and $\lor$-introduction \eqref{B:DisIntro} if needed.) 
By $\land$-elimination \eqref{B:ConElim} and transitivity, our second assumption entails $\vdash_{L} (\phi \land \phi') \to (\chi \lor \psi')$. Since $\vdash (\phi \land \phi') \to \phi$ \eqref{B:ConElim}, we can use adjunction \eqref{B:Adj}, $\land$-introduction  \eqref{B:ConIntro} and modus ponens \eqref{B:MP} to obtain
\[ \vdash_{L} (\phi \land \phi') \to (\phi \land (\chi \lor \psi')) \, .\]
By the distributivity axiom \eqref{B:Distrib} and transitivity:
\begin{equation}\label{eq:step}
\vdash_{L} (\phi \land \phi') \to ((\phi \land \chi) \lor (\phi \land \psi')) \, .
\end{equation}
By $\land$-elimination \eqref{B:ConElim}, $\vdash (\phi \land \psi') \to \psi'$. By $\lor$-introduction \eqref{B:DisIntro} and transitivity, \[\vdash (\phi \land \psi') \to ((\phi \land \chi) \lor \psi') \, . \]
By $\lor$-introduction \eqref{B:DisIntro}, $\vdash (\phi \land \chi) \to ((\phi \land \chi) \lor \psi')$. Hence, adjunction \eqref{B:Adj}, $\lor$-introduction \eqref{B:DisIntro} and modus ponens \eqref{B:MP} yield $\vdash ((\phi \land \chi) \lor (\phi \land \psi')) \to ((\phi \land \chi) \lor \psi')$. By transitivity \and \eqref{eq:step}, we obtain:
\begin{equation}\label{eq:step2}
\vdash_{L} (\phi \land \phi') \to ((\phi \land \chi) \lor \psi') \, .
\end{equation}
Our first assumption $\vdash_{L} (\phi \land \chi) \to \psi$ entails $\vdash (\phi \land \chi) \to (\psi \lor \psi')$ by $\lor$-introduction \eqref{B:DisIntro} and transitivity. By $\lor$-introduction \eqref{B:DisIntro}, adjunction, $\lor$-elimination \eqref{B:DisElim} and modus ponens \eqref{B:MP}, we obtain
\begin{equation}\label{eq:step3}
\vdash ((\phi \land \chi) \lor \psi') \to (\psi \lor \psi') \, .
\end{equation}
Applying transitivity to \eqref{eq:step2} and \eqref{eq:step3} yields:
\[ \vdash_{L} (\phi \land \phi') \to (\psi \lor \psi') \, .\] 
Because $\phi \land \phi'$ is a conjunction of elements of $\Delta$, and $\psi \lor \psi'$ is a disjunction of elements of $\nabla$, this contradicts the assumption that $\langle \Delta, \nabla\rangle$ is $L$-independent.
\end{proof}

\begin{proof}[Proof of Lemma \ref{lem:PExGeneral}]
Assume that the pair $\langle \Delta, \nabla\rangle$ is $L$-independent. Since $Fm_{\mathfrak{L}'}$ is denumerable, we can assume its formulas are enumerated as $\langle \chi_i\rangle_{i = 1}^{\infty}$. We define a sequence of pairs $\langle\Delta_i, \nabla_i\rangle$ from $i=1$ to $\infty$:
\begin{itemize}
\item $\Delta_1 = \Delta$ and $\nabla_1 = \nabla$
\item $\langle \Delta_{i + 1}, \nabla_{i + 1} \rangle = 
	\begin{cases}
	\langle \Delta_i \cup \{ \chi_i \}, \nabla_i \rangle & \text{if } \langle \Delta_i \cup \{ \chi_i \}, \nabla_i\rangle \text{ is } L\text{-independent};\\
	\langle \Delta_i, \nabla_i \cup \{ \chi_i \} \rangle & \text{otherwise.}
	\end{cases}		
	$
\end{itemize}
We define $\Delta_{\infty} = \bigcup_{i = 1}^{\infty} \Delta_i$ and $\nabla_{\infty} = \bigcup_{i = 1}^{\infty} \nabla_i$.

\emph{Claim 1: Pairs $\langle \Delta_i, \nabla_i\rangle$ are $L$-independent for all $i$.} We proceed by induction on $i$. The base case $i = 1$ holds by assumption, and the induction step is established directly using Lemma \ref{lem:StepGeneral}. 

\emph{Claim 2: $\Delta \subseteq \Delta_{\infty}$ and $\nabla \cap \Delta_{\infty} = \emptyset$.} The first part holds by definition. For the second part, suppose for a contradiction that there is a formula $\chi \in \nabla \cap \Delta_{\infty}$. Because $\chi \in \Delta_{\infty}$, we have $\chi \in \Delta_i$ for some $i$. Since $\nabla \subseteq \nabla_i$, we also have $\chi \in \nabla_i$.
By the identity axiom \eqref{B:Id}, the presence of $\chi$ in both sets means the pair $\langle \Delta_i, \nabla_i\rangle$ is not $L$-independent. This contradicts Claim 1.

\emph{Claim 3: $\Delta_{\infty} \cup \nabla_{\infty} = Fm_{\mathfrak{L}'}$.} This follows from the construction. The sequence $\langle \chi_i\rangle_{i = 1}^{\infty}$ enumerates all formulas in $Fm_{\mathfrak{L}'}$, and at each step $i$, the formula $\chi_i$ is added to either $\Delta_{i+1}$ or $\nabla_{i+1}$. 

\emph{Claim 4: $\Delta_{\infty}$ is a prime $L$-theory.} 
Firstly, we show that $\Delta_{\infty}$ is closed under $L$-provable implications. Assume $\vdash_{L} \phi \to \psi$ and $\phi \in \Delta_{\infty}$. If $\psi \notin \Delta_{\infty}$, then $\psi \in \nabla_{\infty}$ by Claim 3. Thus, there exists an index $k$ large enough such that both $\phi \in \Delta_k$ and $\psi \in \nabla_k$. However, since $\vdash_{L} \phi \to \psi$, $\langle \Delta_k, \nabla_k\rangle$ is not $L$-independent, contradicting Claim 1. Therefore, $\psi \in \Delta_{\infty}$.

Secondly, we show that $\Delta_{\infty}$ is closed under conjunction. Assume $\phi, \psi \in \Delta_{\infty}$. This implies there exist indices $i, j$ such that $\phi \in \Delta_i$ and $\psi \in \Delta_j$. Let $k = \max(i, j)$, which ensures $\phi, \psi \in \Delta_k$. If $\phi \land \psi \notin \Delta_{\infty}$, then $\phi \land \psi \in \nabla_{\infty}$, meaning there is some index $m$ such that $\phi \land \psi \in \nabla_m$. By taking $n = \max(k, m)$, we have $\phi, \psi \in \Delta_n$ and $\phi \land \psi \in \nabla_n$. Since $\vdash (\phi \land \psi) \to (\phi \land \psi)$ by \eqref{B:Id}, this contradicts the $L$-independence of $\langle \Delta_n, \nabla_n \rangle$. Thus, $\phi \land \psi \in \Delta_{\infty}$.

Finally, we show that $\Delta_{\infty}$ is prime. Let $\phi \lor \psi \in \Delta_{\infty}$. Suppose for a contradiction that $\phi \notin \Delta_{\infty}$ and $\psi \notin \Delta_{\infty}$. By Claim 3, this implies $\phi, \psi \in \nabla_{\infty}$. Consequently, there exists an index $k$ large enough such that $\phi \lor \psi \in \Delta_k$ and $\phi, \psi \in \nabla_k$. Since $\vdash (\phi \lor \psi) \to (\phi \lor \psi)$ by \eqref{B:Id}, $\langle \Delta_k, \nabla_k \rangle$ is not $L$-independent, which once again contradicts Claim 1. Therefore, either $\phi \in \Delta_{\infty}$ or $\psi \in \Delta_{\infty}$.
\end{proof}

Since $\itAc{BS4}$ is an extension of $B$ in a denumerable language, Lemma \ref{lem:PExGeneral} entails Lemma \ref{lem:PE}. Since intuitionistic propositional logic can be axiomatised by a Hilbert-style system that extends $\mathit{B}$, the system $\itA{iS4}$ is also an extension of $B$ in a denumerable language. Hence, Lemma \ref{lem:PExGeneral} entails Lemma \ref{lem:PEint}.

\subsection{Truth Lemma for relevant implication}\label{app:TruthImp}

Let $L$ be any extension of $B$ with extra axiom schemata and rules, over any language $\Lan'$ extending $\Lan$ such that $Fm_{\Lan'}$ is denumerable. Hence, Lemma \ref{lem:PExGeneral} applies. We write $\vdash \phi$ to indicate that $\phi$ is provable in $L$.
Assume that for each $\f, \ff \in Fm_{\mathfrak{L}}'$ there is a formula $\omega (\f, \ff) \in Fm_{\mathfrak{L}'}$ such that $\vdash \f \to (\ff \to \omega(\f, \ff))$ and a formula $\omega'(\f, \ff)$ such that $\vdash \f \to (\omega'(\f,\ff) \to \ff)$.

We note that such formulas are available in $B$ over an extended language containing \emph{fusion} $\circ$ and \emph{converse implication} $\leftarrow$ since they come with axioms $\f \to (\ff \to (\f \circ \ff))$ and $\f \to ((\ff \leftarrow \f) \to \ff)$. 
However, any extension of $B$ containing the axioms for $\A$ is also sufficient since there we have $\vdash \f \to (\ff \to \E \f)$ and $\vdash \f \to (\neg\E \f \to \ff)$ by \eqref{a:T}, \eqref{a:AecqImpNeg} and \eqref{a:AweakNeg}. 

Let $X$ be the set of all non-empty non-trivial prime $L$-theories. We define a ternary relation $R$ on $X$ by stipulating that $Rxyz$ iff $\forall \f, \ff \in Fm_{\mathfrak{L}'}$, if $\f \to \ff \in x$ and $\f \in y$, then $\ff \in z$.

\begin{lemma}\label{lem:TruthImpGeneral}
$\f \to \ff \in x$ iff $\forall y,z \in X \colon$ if $Rxyz$ and $\f \in y$, then $\ff \in z$.
\end{lemma}
\begin{proof}
The left-to-right implication follows from the definition of $R$. The converse implication is established as follows. Assume that $\f \to \ff \notin x$. We define $y_0 = \{ \alpha \mid \vdash \f \to \alpha \}$. Take the following pair:
\[ \Delta_z = \{ \beta \mid \exists \alpha \in y_0 \colon \alpha \to \beta \in x \} \qquad \nabla_z = \{ \psi \}\]
Note that $\f \in y_0$ and consequently $\Delta_z$ is non-empty: since $\psi \in x$ for some $\psi$ and $\vdash \psi \to (\f \to \omega(\ff, \f))$, we have $\f \to \omega(\ff,\f) \in x$, meaning that $\omega(\ff, \f) \in \Delta_z$. 
We prove that $\langle \Delta_z, \nabla_z\rangle$ is $L$-independent. If not, there is a conjunction $\beta = \beta_1 \land \ldots \land \beta_n$ of elements of $\Delta_z$ such that $\vdash \beta \to \psi$.  For each $\beta_i$ there is $\alpha_i \in y_0$ such that $\alpha_i \to \beta_i \in x$. Let $\alpha = \alpha_1 \land \ldots \land \alpha_n$. Since $\vdash \alpha \to \alpha_i$ for all $i$, we have $\vdash (\alpha_i \to \beta_i) \to (\alpha \to \beta_i)$ using suffixing \eqref{B:Suff}. Hence, $\alpha \to \beta_i \in x$ for all $i$, which means that $\alpha \to \beta \in x$ by repeated application of $\land$-introduction \eqref{B:ConIntro}. 
Since $\vdash \beta \to \psi$, we obtain $\vdash (\alpha \to \beta) \to (\alpha \to \psi)$ by suffixing \eqref{B:Suff}. Hence, $\alpha \to \psi \in x$. 
Since $\alpha_i \in y_0$ for all $i$, we have $\vdash \f \to \alpha_i$, which implies $\vdash \f \to \alpha$ using adjunction \eqref{B:Adj} and $\land$-introduction \eqref{B:ConIntro}. By suffixing \eqref{B:Suff}, $\vdash (\alpha \to \psi) \to (\phi \to \psi)$. Because $\alpha \to \psi \in x$, we obtain $\phi \to \psi \in x$, contradicting our assumption. 
Hence, $\langle \Delta_z, \nabla_z\rangle$ must be $L$-independent. Since both $\Delta_z$ and $\nabla_z$ are non-empty, Lemma \ref{lem:PExGeneral} implies that there is a non-empty non-trivial prime $L$-theory $z$ extending $\Delta_z$ and disjoint from $\nabla_z$.

Now consider the following pair:
\[ \Delta_y = \{ \f \} \qquad \nabla_y = \{ \gamma \mid \exists \delta \notin z \colon \gamma \to \delta \in x \}\]
Note that $\chi \in x$ and $\delta \notin z$ for some $\chi, \delta$. Since $\vdash \chi \to (\omega'(\chi, \delta) \to \delta)$, we have $\omega'(\chi, \delta) \to \delta \in x$. Hence, $\omega'(\chi, \delta) \in \nabla_y$. 
We prove that $\langle \Delta_y, \nabla_y\rangle$ is $L$-independent. If not, then there is a disjunction $\gamma = \gamma_1 \lor \ldots \lor \gamma_m$ such that $\vdash \phi \to \gamma$. For each $i$, there is $\delta \notin z$ such that $\gamma_i \to \delta_i \in x$. We have $\vdash \delta_i \to \delta$ for all $i$ and so, by prefixing \eqref{B:Pref}, $\vdash (\gamma_i \to \delta_i) \to (\gamma_i \to \delta)$. Since $\gamma_i \to \delta_i \in x$ for all $i$, we obtain $\gamma_i \to \delta$ for all $i$. By closure under conjunction and $\lor$-introduction \eqref{B:DisIntro}, $\gamma \to \delta \in x$. 
Since we assume that $\vdash \phi \to \gamma$, we obtain $\vdash (\gamma \to \delta) \to (\phi \to \delta)$ by suffixing \eqref{B:Suff}. Hence, $\phi \to \delta \in x$. 
We know that $\phi \in y_0$ and so $\delta \in \Delta_z$. Since $\Delta_z \subseteq z$, we have $\delta \in z$. This means that $\delta_i \in z$ for some $i$, contradicting our assumption. 
Hence, $\langle \Delta_y, \nabla_y\rangle$ is $L$-independent. Since both $\Delta_y$ and $\nabla_y$ are non-empty, Lemma \ref{lem:PExGeneral} implies that there is a non-empty non-trivial prime $L$-theory $y$ extending $\Delta_y$ and disjoint from $\nabla_y$. 

By construction of $y$ and $z$, we know that $\f \in y$ and $\ff \notin z$. To show that $Rxyz$, assume $\alpha \to \beta \in x$ and $\beta \notin z$. Then $\alpha \notin y$ by the definition of $\nabla_y$, as desired.  
\end{proof}

%\nocite{*}
\bibliographystyle{eptcs}
\bibliography{itel}

\end{document}